\documentclass{article}%
\usepackage{graphicx}
\usepackage{amsmath}
\usepackage{amsfonts}
\usepackage{amssymb}%
\providecommand{\U}[1]{\protect\rule{.1in}{.1in}}
\numberwithin{equation}{section}
\newtheorem{theorem}{Theorem}

\newtheorem{lemma}[theorem]{Lemma}

\newtheorem{proposition}[theorem]{Proposition}
\newtheorem{remark}[theorem]{Remark}

\newenvironment{proof}[1][Proof]{\textbf{#1.} }{\ \rule{0.5em}{0.5em}}
\begin{document}

\bigskip

\bigskip

\bigskip

\begin{center}
\textbf{A CLASS OF DUST-LIKE SELF-SIMILAR SOLUTIONS OF THE MASSLESS
EINSTEIN-VLASOV SYSTEM.}

\bigskip

\bigskip Alan D. Rendall\footnote{Max Planck Institute for Gravitational
Physics, Albert Einstein Institute, Am M\"{u}hlenberg 1, 14476 Potsdam,
Germany}, Juan J. L. Vel\'{a}zquez\footnote{ICMAT
(CSIC-UAM-UC3M-UCM),\ Universidad Complutense, Madrid 28035, Spain.}

\bigskip
\end{center}

\begin{abstract}
In this paper the existence of a class of self-similar solutions of the
Einstein-Vlasov system is proved. The initial data for these solutions are
not smooth, with their particle density being supported in a submanifold of
codimension one. They can be thought of as intermediate between smooth
solutions of the Einstein-Vlasov system and dust. The motivation for
studying them is to obtain insights into possible violation of weak 
cosmic censorship by solutions of the Einstein-Vlasov system. By assuming
a suitable form of the unknowns it is shown that the existence question
can be reduced to that of the existence of a certain type of solution of a 
four-dimensional system of ordinary differential equations depending on two 
parameters. This solution starts at a particular point $P_0$ and converges to 
a stationary solution $P_1$ as the independent variable tends to infinity. 
The existence proof is based on a shooting argument and involves relating 
the dynamics of solutions of the four-dimensional system to that of solutions 
of certain two- and three-dimensional systems obtained from it by limiting 
processes.

\end{abstract}

\bigskip

\section{INTRODUCTION}

\bigskip

\bigskip

\bigskip

It is well known that solutions of the Einstein equations coupled with
suitable models of matter can yield singularities in finite time. The unknowns
in these equations are the spacetime metric and some matter fields. The exact
nature of the latter depends on the physical situation being considered. The
usual terminology in general relativity is that there is said to be a
singularity if the metric fails to be causally geodesically complete, i.e. if
there are timelike or null geodesics which in at least one direction are
inextendible and of finite affine length. The singularity is said to be in the
future or the past according to the incomplete direction of the geodesics. It
is expected on the basis of physical intuition, and known to be true in some
simple cases, that the geodesic incompleteness is associated with the energy
density or some curvature invariants blowing up. For background on this
subject see textbooks such as \cite{hawking}, \cite{wald} and \cite{rendall08}%
. One of the best known types of singularities in general relativity are those
which occur inside black holes. When a black hole is formed by the collapse of
matter it is known that under suitable circumstances an event horizon is
formed which ensures that the singularity can have no influence on distant observers.

Mathematical relativity is the study of the properties of solutions of the
Einstein equations coupled to various matter equations. One of the main
questions in the field is the cosmic censorship hypothesis. There are two
versions of this conjecture called weak and strong cosmic censorship, both of
which were proposed by Roger Penrose. It should be noted that, contrary to
what the names might suggest, the strong version does not imply the weak one.
The results proved in what follows are motivated by weak cosmic censorship and
strong cosmic censorship will not be discussed further here. Weak cosmic
censorship is a statement which concerns isolated systems in general
relativity. Mathematically this means considering solutions of the Einstein
equations which evolve from asymptotically flat initial data. Initial data for
the Einstein equations consist of a Riemannian metric $h_{ab}$, a symmetric
tensor $k_{ab}$ and some matter fields which for the moment will be denoted
generically by $F_{0}$, all defined on a three-dimensional manifold $S$.
Solving the Cauchy problem for the Einstein-matter equations means embedding
the manifold $S$ into a four-dimensional manifold $M$ on which are defined a
Lorentzian metric $g_{\alpha\beta}$ and matter fields $F$ such that $h_{ab}$
and $k_{ab}$ are the pull-backs to $S$ of the induced metric and second
fundamental form of the image of the embedding of $S$ while $F_{0}$ is the
pullback of the matter fields. The metric $g_{\alpha\beta}$ and the matter
fields $F$ are required to satisfy the Einstein-matter equations. A
comprehensive treatment of the Cauchy problem for the Einstein equations can
be found in \cite{ringstrom09}. Initial data on ${\mathbb{R}}^{3}$ are called
asymptotically flat if the metric $h_{ab}$ tends to the flat metric at
infinity in a suitable sense while $k_{ab}$ and $F_{0}$ tend to zero.
Physically this corresponds to concentrating attention on a particular
physical system while ignoring the influence of the rest of the universe.

A solution of the Einstein-matter equations evolving from initial data is said
to be a \textit{development} of that data if each inextendible causal curve
intersects the initial hypersurface precisely once. When this property holds
the initial hypersurface is said to be a Cauchy hypersurface for that
solution. In general, a solution is called globally hyperbolic if it admits a
Cauchy hypersurface. For prescribed data there is a development which is
maximal in the sense that any other development can be embedded into it. It is
unique up to a diffeomorphism which preserves the initial hypersurface.

In a spacetime evolving from asymptotically flat data it is often possible to
define future null infinity ${\mathcal{I}}^{+}$ as a set of ideal endpoints of
complete future-directed null geodesics. We can say that any singularity
occurring does not influence events near infinity if there is no inextendible
causal curve to the future of the initial hypersurface which is incomplete in
the past while intersecting a future-complete null geodesic. The first of
these properties means intuitively that this curve represents a signal which
comes out of a singularity while the second property means that it reaches a
region which can communicate with infinity. If a curve of this type does exist
it is said that a globally naked singularity exists. The past of null
infinity, $J^{-}({\mathcal{I}}^{+})$, is the set of points for which there is
a future-directed causal curve starting there and going to null infinity. The
complement of $J^{-}({\mathcal{I}}^{+})$ is called the black hole region. Its
boundary is called the event horizon and is a null hypersurface in $M$.

There is a notion of completeness of null infinity. A precise definition will
not be given here but roughly speaking it corresponds to the situation where
there are timelike curves contained in $J^{-}({\mathcal{I}}^{+})$ which exist
for a infinite time towards the future. Physically this means that there are
observers which can remain outside the black hole for an unlimited amount of
time. If the maximal globally hyperbolic development of asymptotically flat
initial data always has a complete null infinity then this ensures the absence
of globally naked singularities. For any inextendible causal curve to the
future of the initial surface which goes to null infinity must intersect the
initial hypersurface. Hence it cannot be incomplete in the past. The
completeness of ${\mathcal{I}}^{+}$ ensures that the solution is large enough
to represent the whole future of a system evolving from the initial data under
consideration. The intuitive content of the weak cosmic censorship hypothesis
is that in the time evolution corresponding to initial data for the Einstein
equations coupled to reasonable (non-pathological) matter the existence of a
singularity implies that of an event horizon which covers the singularity and
hides it from distant observers. Often this is weakened to the requirement
that a horizon exists in the case of generic initial data. Up to now this
intuitive picture has only been developed into a precise mathematical
formulation under special circumstances. In general finding the correct
formulation is part of the problem to be solved.

Due to the mathematical complexity of the Einstein equations many of the
studies related to singularity formation for these equations have been carried
out for spherically symmetric solutions. In spherical symmetry the Einstein
vacuum equations are non-dynamical due to Birkhoff's theorem, which says that
any spherically symmetric vacuum solution is locally isometric to the
Schwarzschild solution and, in particular, static. Thus it is essential to
include matter of some kind. A matter model which has proved very useful for
this task is the scalar field. This is a real-valued function $\phi$ which
satisfies the wave equation $\nabla^{\alpha}\nabla_{\alpha}\phi=0$. In this
case the Einstein equations take the form $R_{\alpha\beta}=8\pi\nabla_{\alpha
}\phi\nabla_{\beta}\phi$, where $R_{\alpha\beta}$ is the Ricci curvature of
$g_{\alpha\beta}$. The spherically symmetric Einstein-scalar field equations
were studied in great detail in a series of papers by Demetrios Christodoulou.
This culminated in \cite{christodoulou94} and \cite{christodoulou99}. In
\cite{christodoulou94} it was shown that in this system naked singularities
can evolve from regular asymptotically flat initial data. This represents a
problem for the weak cosmic censorship hypothesis but the conjecture can be
saved by a genericity assumption since it was shown in \cite{christodoulou99}
that generic initial data do not lead to naked singularities.

For the spherically symmetric Einstein-scalar field equations it is known from
the work of Christodoulou \cite{christodoulou86} that small asymptotically 
flat initial data lead to a solution which is geodesically complete and hence 
free of singularities. (In fact this small data result has recently been 
extended to the case without symmetry \cite{lindblad}.) On
the other hand there are certain large initial data for which it is known that
a black hole is formed. The threshold between these two types of behaviour was
studied in influential work by Choptuik \cite{choptuik} and many other papers
since. This area of research is known as critical collapse and is surveyed in
\cite{gundlach07}. It is entirely numerical and heuristic and unfortunately
mathematically rigorous results are not yet available.

The scalar field provides a simple and well-behaved matter model. At the same
time no such field has been experimentally observed and the matter fields of
importance for applications to astrophysics are of other kinds. One
astrophysically relevant matter field which has good mathematical properties
is collisionless matter described by the Vlasov equation. The necessary
definitions are given in the next section. For the moment let it just be noted
that the unknown in the Vlasov equation is a non-negative real-valued function
$f(t,x^{a},v^{b})$ depending on local coordinates $(t,x^{a})$ on $M$ and
velocity variables $v^{b}$. Analogues of a number of the results proved for
the scalar field have been proved for the Einstein-Vlasov system. For small
initial data the solutions are geodesically complete \cite{rein92}. There are
certain large initial data for which a black hole is formed
\cite{andreasson08}. The threshold between these two types of behaviour has
been investigated numerically in \cite{rein98} and \cite{olabarrieta}. A
closely related matter model which has been very popular in theoretical
general relativity is dust, a fluid with vanishing pressure. It is equivalent
to consider distributional solutions of the Vlasov equation of the form
$f(t,x^{a},v^{b})=\rho(t,x^{a})\delta(v^{b}-u^{b}(t,x^{a}))$ where the
$\delta$ is a Dirac distribution. From many points of view dust is relatively
simple to analyse. Unfortunately it has a strong tendency to form
singularities where the energy density blows up, even in the absence of
gravity. For this reason it must be regarded as pathological and of limited
appropriateness for the investigation of cosmic censorship. A detailed
mathematical study of formation of singularities in the Einstein equations
coupled to dust was given in \cite{christodoulou84}. In spherical symmetry
dust particles move as spherical shells. It can easily happen that shells
including a strictly positive total mass come together at one radius and 
this causes the density to blow up. This effect is known as shell-crossing.

The motivation for this paper is the wish to understand cosmic censorship
better for spherically symmetric solutions of the Einstein-Vlasov system. Is
it true that in asymptotically flat spherically symmetric solutions of the
Einstein-Vlasov system there are no naked singularities for generic data so
that collisionless matter is as well-behaved as the scalar field? Could it
even be that the Vlasov equation is better-behaved and that there are no naked
singularities at all? No answers to these questions, positive or negative, are
available although considerable effort has been invested into obtaining a
positive answer. In what follows we try to obtain new insights by approaching
a negative result through an interpolation between dust and smooth solutions
of the Vlasov equation and looking for self-similar solutions. There are some
results on related equations which give some hints. In the case of the
Vlasov-Poisson system, the non-relativistic analogue of the Einstein-Vlasov
system, global existence for general data, not necessarily symmetric, was
proved by Pfaffelmoser \cite{pfaffelmoser} and Lions and Perthame
\cite{lions}. The relativistic Vlasov-Poisson system, which is in some sense
intermediate between the Vlasov-Poisson and Einstein-Vlasov systems, (but not
in all ways) has been shown to have solutions which develop singularities in
finite time. Rather precise information is available about the nature of these
singularities \cite{lemou}.

As a side remark, we mention a paper \cite{shapiro} where it was suggested
that naked singularities are formed in solutions of the Einstein-Vlasov
system. The solutions concerned were axially symmetric but not spherically
symmetric. The work is purely numerical but trying to understand what it means
for the analytical problem leads to the conclusion that the solutions computed
in \cite{shapiro} were dust solutions rather than smooth solutions of the
Einstein-Vlasov system. This is discussed in \cite{rendall92}. There are also
reasons for doubting that the numerical results really show the formation of a
naked singularity \cite{wald91}.

A class of distributional solutions of the Einstein-Vlasov system intermediate
between smooth solutions and dust is given by the Einstein clusters
\cite{einstein}. These are spherically symmetric and static, i.e. there exists
a timelike Killing vector field which is orthogonal to spacelike
hypersurfaces. It is supposed that the support of $f$ consists of $v^{a}$ such
that the geodesics with these initial data are tangent to the spheres of
constant distance from the centre of symmetry on these spacelike
hypersurfaces. This means that the radial velocity and its time derivative in
the geodesic equation are zero. These are in general two independent
conditions on the data at a given time. A wider class, the generalized
Einstein clusters \cite{datta}, \cite{bondi}, is obtained as follows. In the
case of the Einstein clusters taking the union of the spheres at a fixed
distance from the centre defines a \ foliation of the spacetime by timelike
hypersurfaces and the condition on the support means that the four-velocity of
a particle with the given initial data is everywhere tangent to these timelike
hypersurfaces. The generalized Einstein clusters are obtained by dropping the
condition of staticity and replacing the family of timelike hypersurfaces
invariant under the timelike Killing vector field by another foliation by
timelike hypersurfaces which intersect any Cauchy surface in spheres and whose
equation of motion follows from the Vlasov equation. Once again the
four-velocity of a particle in the support of $f$ is tangent to these
hypersurfaces at all times. An analytical formulation of this definition will
be given in the next section. It should be noted that the generalized Einstein
clusters exhibit shell-crossing singularities and thus can still be thought of
as pathological. We are interested in them as an intermediate step towards
better-behaved matter models.

There are two major differences between the generalized Einstein clusters and
the solutions studied in this paper. In the case of Einstein clusters the
value of the angular momentum of the particles $F$ is uniquely determined by
the distance $r$ to the centre of symmetry. By contrast, in the solutions
studied in this paper the angular momentum takes a continuous range of values
for each value of $r.$ The second difference is that in the case of Einstein
clusters at each spacetime point the component of the velocity vector $v^{a}$
of a particle in the direction of the vector $\partial_{r}$ takes on only one
value. In the case of the solutions obtained in this paper the component along
the direction $\partial_{r}$ of the velocity vector takes on two different
values at most spacetime points. This only fails at some exceptional values of
$r$ at a given time. This difference in the structure of the generalized
Einstein clusters and the solutions considered in this paper is what gives
some plausibility to the idea that the solutions described here could be a big
step towards better-behaved matter models. From the physical point of view, in
the case of the generalized Einstein clusters the material particle with the
smallest value of $r$ would not experience any gravitational field, and
therefore could not approach the centre $r=0$ unless its angular momentum
vanished. In the solutions studied in this paper, since two radial velocities
are allowed at each spacetime point, the material particle with the smallest
value of $r$ changes in time. This allows the occurrence of a collective
collapse of the whole distribution of particles towards the origin with some
of them coming closer and closer to the center as the value of some suitable
time coordinate $t$ increases.

Self-similar solutions of the massless Einstein-Vlasov system have also been
considered in the paper \cite{martingarcia}. There are several differences
between the approach in \cite{martingarcia} and the one considered in this
paper. The first one is the choice of the rescaling group under which the
solutions are invariant. The massless Einstein-Vlasov system is invariant
under a two-dimensional group of rescalings. The choice of a particular
one-dimensional rescaling group has been made in this paper by imposing that
the distribution function $f$ for the particles remains always of order one
(see Sections \ref{Section1}, \ref{Section2}). This condition is natural,
because the function $f$ is invariant along characteristic curves. On the
contrary, the choice of one-dimensional rescaling group for the solutions in
\cite{martingarcia} imposes that $f$ becomes unbounded near the singularity
for the particles within the self-similar region, something that can be
achieved assuming that the distribution of matter is singular near the
light-cone. The second difference between the solutions in \cite{martingarcia}
and those in this paper is that the solutions in \cite{martingarcia} can be
thought of as self-similar perturbations of the flat Minkowski space. As a
matter of fact they have been computed by means of a perturbative iteration
procedure that takes flat space as a starting point and where the terms in the
resulting series have been computed numerically. By contrast, the solutions of
this paper are obtained by means of a shooting procedure in which a parameter
that measures the amount of energy in the self-similar region is of order one.
The approach in this paper uses purely analytical methods and does not rely on
numerical computations. On the other hand, in order to simplify the arguments, 
we have restricted the analysis in this paper to the study of dust-like
solutions, an assumption that was not made in \cite{martingarcia}.

The plan of the paper is as follows. We will first reduce the problem of
finding self-similar solutions of the Einstein-Vlasov system to an ODE problem
that can be transformed into a four-dimensional system using suitable changes
of variables. Using these transformations it will be seen that the
construction of the desired self-similar solutions reduces to finding a
particular orbit in the corresponding four-dimensional space connecting a
certain point with a steady state that has a three-dimensional stable
manifold. The existence of such an orbit will be shown by adjusting a
parameter that measures the density of particles in a particular perturbative
limit. The precise limit under consideration, which has the goal of making the
problem feasible using analytical methods, corresponds to assuming that the
radius of the region empty of particles, measured in the natural self-similar
variables, is small.

\section{\label{Section1}THE EINSTEIN-VLASOV SYSTEM IN\protect\linebreak%
\ SCHWARZSCHILD COORDINATES.}

\bigskip

We do not use exactly the classical Schwarzschild coordinates, but a slight
modification of them that normalizes the time to be the proper time at the
center $r=0$. The metric is given by (cf. \cite{Rein}):
\begin{equation}
ds^{2}=-e^{2\mu\left(  t,r\right)  }dt^{2}+e^{2\lambda\left(  t,r\right)
}dr^{2}+r^{2}\left(  d\theta^{2}+\sin^{2}\theta d\varphi^{2}\right)  .
\label{met1}%
\end{equation}
If we restrict our attention to spherically symmetric solutions it is
convenient to use the quantities (cf. \cite{Rein}):%
\[
r=\left\vert x\right\vert \;,\;w=\frac{x\cdot v}{r}\;,\;F=\left\vert x\wedge
v\right\vert ^{2}%
\]
to parametrize the velocity variables. In particular $F$ is constant along
characteristics. Writing the particle density as
\[
f=f\left(  r,w,F,t\right)
\]
the Einstein-Vlasov system for spherically symmetric solutions in these
coordinates becomes:
\begin{equation}
\partial_{t}f+e^{\mu-\lambda}\frac{w}{E}\partial_{r}f-\left(  \lambda
_{t}w+e^{\mu-\lambda}\mu_{r}E-e^{\mu-\lambda}\frac{F}{r^{3}E}\right)
\partial_{w}f=0 \label{S1E1}%
\end{equation}
where:
\begin{equation}
E=\sqrt{1+w^{2}+\frac{F}{r^{2}}} \label{S1E2}%
\end{equation}
and the functions $\lambda,\;\mu$ that characterize the gravitational field
satisfy:
\begin{align}
e^{-2\lambda}\left(  2r\lambda_{r}-1\right)  +1  &  =8\pi r^{2}\rho
,\label{S1E3}\\
e^{-2\lambda}\left(  2r\mu_{r}+1\right)  -1  &  =8\pi r^{2}p \label{S1E4}%
\end{align}
with boundary conditions:
\begin{align}
\mu\left(  0\right)   &  =0\;\;,\;\lambda\left(  0\right)  =0,\;\label{S1E5}\\
\lambda\left(  \infty\right)   &  =0. \label{S1E6}%
\end{align}
On the other hand $\rho$ and $p$ are given by:
\begin{align}
\rho &  =\rho\left(  r,t\right)  =\frac{\pi}{r^{2}}\int_{-\infty}^{\infty
}\left[  \int_{0}^{\infty}EfdF\right]  dw,\label{S1E7}\\
p  &  =p\left(  r,t\right)  =\frac{\pi}{r^{2}}\int_{-\infty}^{\infty}\left[
\int_{0}^{\infty}\frac{w^{2}}{E}fdF\right]  dw. \label{S1E8}%
\end{align}

With these basic equations in hand it is possible to give some details
concerning generalized Einstein clusters, as promised in the introduction.
These are not required to understand the main results of the paper but help to
put those results into a wider context. A distributional solution of the
Vlasov equation whose support is a smooth submanifold $\Sigma$ has the
property that $\Sigma$ is a union of characteristics of the equation. A simple
example is that of dust where the support is the graph of a function
$u^{a}(t,x)$ of the form $W(t,r)\frac{x^{a}}{r}$. When expressed in terms of
polar coordinates this becomes the graph of a function $W(t,r)$ augmented by
the condition $F=0$. Here the function $W$ solves the equations
\begin{align}
&  \frac{dR}{dt}=e^{\mu(t,R)-\lambda(t,R)}\frac{W}{E},\\
&  \frac{dW}{dt}=-(\dot\lambda(t,R)W+e^{\mu(t,R)-\lambda(t,R)}\mu^{\prime
}(t,R)E)
\end{align}
where $E=\sqrt{1+W^{2}}$.

Now consider the generalized Einstein clusters. They are only defined under
the condition of spherical symmetry. They can be thought of as defining a
matter model which can be used in the spherically symmetric Einstein-matter
equations. Here they will be described in terms of Schwarzschild coordinates.
The basic unknown is a function $R(t,r)$ which satisfies $R(0,r)=r$. It is the
area radius at time $t$ of the shell which had area radius $r$ at time $0$. As
input we require a function $F(r)$ which is the angular momentum of the
particles on the shell which was at radius $r$ at time zero and $N(r)$ which
is the density of particles per shell evaluated on the shell which had area
radius $r$ at $t=0$. For some purposes it is more convenient to use $R$ as a
radial coordinate instead of $r$ and this is what was done in the original
papers \cite{datta} and \cite{bondi}. For a given shell at a given time the
angular momentum and radial velocity of the particles are fixed and so the
intersection of the support of the solution with the fibre of the mass shell
over the point with coordinates $(t,r)$ has codimension two. The following
equations should be satisfied:
\begin{align}
&  \frac{dR}{dt}=e^{\mu-\lambda}\frac{W}{E},\\
&  \frac{dW}{dt}=-\left(  \dot{\lambda}W+e^{\mu-\lambda}\mu^{\prime}E
-e^{\mu-\lambda}\frac{F}{R^{3}E}\right)
\end{align}
where $E=\sqrt{1+W^{2}+\frac{F}{R^{2}}}$ and the functions $\lambda$ and $\mu$
are to be evaluated at the point $(t,R)$. These are the full characteristic
equations for the Vlasov equation. The difference in the coupled system comes
from the fact that the expressions for the components of the energy-momentum
tensor are different in the two cases. In the case of Einstein clusters the
characteristics of interest have $W=0$ and $\frac{dW}{dt}=0$. It follows
immediately that $\frac{dR}{dt}=0$. In that case the angular momentum is
related to the geometry by the relation $F=\frac{r^{3}\mu^{\prime}}%
{1-r\mu^{\prime}}$.

\bigskip

The equations which have been written up to now describe particles of unit
mass. We are interested in the construction of solutions of (\ref{S1E1}%
)-(\ref{S1E8}) supported in a region where $\left(  w^{2}+\frac{F}{r^{2}%
}\right)  $ takes large values near the formation of the singularity. This
suggests replacing (\ref{S1E2}) by:
\begin{equation}
E=\sqrt{w^{2}+\frac{F}{r^{2}}}. \label{S1E9}%
\end{equation}
The system (\ref{S1E1}), (\ref{S1E3})-(\ref{S1E8}), (\ref{S1E9}) is invariant
under the rescaling:
\begin{equation}
r\rightarrow\theta r\;\;,\;\;t\rightarrow\theta t\;\ \ \text{for\ }%
t<0\;,\;\;w\rightarrow\frac{1}{\sqrt{\theta}}w\;\;,\;\;F\rightarrow\theta F
\label{S1E10}%
\end{equation}
for any $\theta>0.$ It is then natural to look for solutions of (\ref{S1E1}),
(\ref{S1E3})-(\ref{S1E9}) invariant under the rescaling (\ref{S1E10}). They
will be the self-similar solutions in which we will be interested in this paper.

The system obtained when (\ref{S1E2}) is replaced by (\ref{S1E9}) can be
interpreted as describing particles of zero rest mass. The rationale for this
assumption is that near the singularity the derived solution will satisfy
$w^{2}+\frac{F}{r^{2}}>>1,$ and therefore it could be expected that it is 
possible to treat the whole Einstein-Vlasov system with massive particles as a
perturbation of the massless problem.

In what follows we will consider solutions of (\ref{S1E1}), (\ref{S1E3}%
)-(\ref{S1E9}) where $f$ is not a bounded function, but a measure concentrated
on some hypersurfaces that will be described in detail later. As was mentioned
in the introduction there is a class of distributional solutions of the
Einstein-Vlasov system which are equivalent to what is usually known in the
literature as dust. From this point of view the solutions considered in this
paper are intermediate between dust and smooth solutions and hence will be
called dust-like solutions. Note, however, that in contrast to dust they do
have some velocity dispersion. The dimension of the support of $f$ in the
tangent space at a given spacetime point is zero for dust, one for generalized
Einstein clusters, two for the solutions in this paper and three for smooth
solutions. For the solutions here it will be possible to describe the
distribution of velocities for the particles at a given point using a function
depending on one coordinate, while a general distribution of velocities
compatible with the assumption of spherical symmetry would depend on two coordinates.

\bigskip

\section{\label{Section2}SELF-SIMILAR SOLUTIONS}

\bigskip

In this section we formulate the system of equations satisfied by the
solutions of (\ref{S1E1}), (\ref{S1E3})-(\ref{S1E8}), (\ref{S1E9}) that are
invariant under the transformation (\ref{S1E10}). We will call these
self-similar solutions in what follows. It is convenient, as a first step, in
order to transform (\ref{S1E1}), (\ref{S1E3})-(\ref{S1E9}) to a more
convenient form to define a new variable:
\begin{equation}
v=\frac{w}{\sqrt{F}}. \label{S3E1}%
\end{equation}
We will assume in the rest of the paper that $f=0$ for $\left(
r,v,F,t\right)  =\left(  r,v,0,t\right)  $ in order to avoid singularities in
(\ref{S3E1}). Moreover, we can even assume a more stringent condition on $f,$
namely $f=0$ for $0\leq F\leq\delta_{0}$ for some $\delta_{0}>0.$ Concerning
the support in the $r$ coordinate, the solutions constructed in this paper
will vanish for $r\leq y_{0}\left(  -t\right)  $ for some $y_{0}>0.$

Making the change of variables $\left(  r,w,F,t\right)  \rightarrow\left(
r,v,F,t\right)  $ and denoting the new distribution function by $f$ with a
slight abuse of notation we can transform the system (\ref{S1E1}),
(\ref{S1E7})-(\ref{S1E9}) into:%

\begin{align}
&  \partial_{t}f+e^{\mu-\lambda}\frac{v}{\tilde{E}}\partial_{r}f-\left(
\lambda_{t}v+e^{\mu-\lambda}\mu_{r}\tilde{E}-e^{\mu-\lambda}\frac{1}%
{r^{3}\tilde{E}}\right)  \partial_{v}f =0,\label{S3E2}\\
&  \tilde{E} =\sqrt{v^{2}+\frac{1}{r^{2}}},\label{S3E3}\\
&  \rho=\frac{\pi}{r^{2}}\int_{-\infty}^{\infty}\tilde{E}\left[  \int
_{0}^{\infty}fFdF\right]  dv,\label{S3E4}\\
&  p =\frac{\pi}{r^{2}}\int_{-\infty}^{\infty}\frac{v^{2}}{\tilde{E}}\left[
\int_{0}^{\infty}fFdF\right]  dv. \label{S3E5}%
\end{align}
Notice that the change of variables (\ref{S3E1}) eliminates the dependence on
the variable $F$ for the characteristic curves associated to the Vlasov
equation (cf. (\ref{S3E2})). Moreover, the functions $\rho$ and $p$ and
therefore the functions $\lambda,\;\mu$ characterizing the gravitational
fields depend on $f$ only through the reduced distribution function:
\begin{equation}
\zeta\left(  r,v,t\right)  \equiv\int_{0}^{\infty}fFdF. \label{S3E6}%
\end{equation}
In particular, it is possible to write a closed problem for the reduced
distribution function that can be obtained multiplying (\ref{S3E2}) by $F$ and
integrating with respect to this variable:
\begin{align}
&  \partial_{t}\zeta+e^{\mu-\lambda}\frac{v}{\tilde{E}}\partial_{r}%
\zeta-\left(  \lambda_{t}v+e^{\mu-\lambda}\mu_{r}\tilde{E}-e^{\mu-\lambda
}\frac{1}{r^{3}\tilde{E}}\right)  \partial_{v}\zeta=0,\label{S3E7}\\
&  \tilde{E} =\sqrt{v^{2}+\frac{1}{r^{2}}},\label{S3E8}\\
&  \rho=\frac{\pi}{r^{2}}\int_{-\infty}^{\infty}\tilde{E}\zeta dv,\label{S3E9}%
\\
&  p =\frac{\pi}{r^{2}}\int_{-\infty}^{\infty}\frac{v^{2}}{\tilde{E}}\zeta dv.
\label{S3E10}%
\end{align}
The system (\ref{S3E7})-(\ref{S3E10}) complemented with (\ref{S1E3}),
(\ref{S1E4}) is a closed system of equations.

We will now study the class of self-similar solutions of the system
(\ref{S1E3}), (\ref{S1E4}), (\ref{S3E2})-(\ref{S3E5}). These are the functions
having the functional dependence:
\begin{align}
f\left(  r,v,F,t\right)   &  =G\left(  y,V,\Phi\right)  \;\;,\;\;\mu\left(
r,t\right)  =U\left(  y\right)  \;\;,\;\;\lambda\left(  r,t\right)
=\Lambda\left(  y\right)  ,\label{S3E11}\\
y  &  =\frac{r}{\left(  -t\right)  }\;\;,\;\;V=\left(  -t\right)
v\;\;,\;\;\Phi=\frac{F}{\left(  -t\right)  }.\;\; \label{S3E12}%
\end{align}
The solutions of (\ref{S1E3}), (\ref{S1E4}), (\ref{S3E2})-(\ref{S3E5}) with
this functional dependence satisfy:
\begin{align}
&  yG_{y}-VG_{V}+\Phi G_{\Phi}+e^{U-\Lambda}\frac{V}{\hat{E}}G_{y}\nonumber\\
&  -\left(  y\Lambda_{y}V+e^{U-\Lambda}U_{y}\hat{E}-e^{U-\Lambda}\frac
{1}{y^{3}\hat{E}}\right)  G_{V}\nonumber\\
&  =0 \label{S4E9}%
\end{align}
where:
\begin{equation}
\hat{E}=\sqrt{V^{2}+\frac{1}{y^{2}}} \label{S4E2}%
\end{equation}
and
\begin{align}
e^{-2\Lambda}\left(  2y\Lambda_{y}-1\right)  +1  &  =8\pi y^{2}\tilde{\rho
},\label{S4E3}\\
e^{-2\Lambda}\left(  2yU_{y}+1\right)  -1  &  =8\pi y^{2}\tilde{p}
\label{S4E4}%
\end{align}
with boundary conditions:
\begin{equation}
U=0\;\;,\;\Lambda=0\;\;\;\text{at\ \ }y=0. \label{S4E6}%
\end{equation}
Here:
\begin{align}
\tilde{\rho}  &  =\frac{\pi}{y^{2}}\int_{-\infty}^{\infty}\hat{E}\left[
\int_{0}^{\infty}G\Phi d\Phi\right]  dV,\label{S4E7}\\
\tilde{p}  &  =\frac{\pi}{y^{2}}\int_{-\infty}^{\infty}\frac{V^{2}}{\hat{E}%
}\left[  \int_{0}^{\infty}G\Phi d\Phi\right]  dV. \label{S4E8}%
\end{align}

\bigskip

The function $G$ which is a solution of (\ref{S4E9})-(\ref{S4E8}) is constant
along the characteristic curves of (\ref{S4E9}) which are given by:
\begin{align}
\frac{dy}{d\sigma}  &  =y+e^{U-\Lambda}\frac{V}{\sqrt{V^{2}+\frac{1}{y^{2}}}%
}=y+e^{U-\Lambda}\frac{Vy}{\sqrt{V^{2}y^{2}+1}},\label{S4E10}\\
\frac{dV}{d\sigma}  &  =-V-\left(  y\Lambda_{y}V+\frac{e^{U-\Lambda}U_{y}}%
{y}\sqrt{V^{2}y^{2}+1}-e^{U-\Lambda}\frac{1}{y^{2}\sqrt{V^{2}y^{2}+1}}\right)
,\label{S4E11}\\
\frac{d\Phi}{d\sigma}  &  =\Phi. \label{S4E12}%
\end{align}
In these equations $\sigma$ is just a parameter that is used to parametrize
the characteristic curves. Its precise definition will be given later in some
specific cases.

\bigskip

The equations (\ref{S4E10})-(\ref{S4E12}) can be integrated explicitly for any
pair of functions $U=U\left(  y\right)  ,\;\Lambda=\Lambda\left(  y\right)  $.
Indeed, the first two equations can be rewritten as:
\begin{align}
\frac{dy}{d\sigma}  &  =e^{-\Lambda}\frac{\partial H}{\partial V}%
,\label{S5E1}\\
\frac{dV}{d\sigma}  &  =-e^{-\Lambda}\frac{\partial H}{\partial y}
\label{S5E2}%
\end{align}
where:
\begin{equation}
H\equiv\frac{e^{U}}{y}\sqrt{V^{2}y^{2}+1}+yVe^{\Lambda}. \label{S5E3}%
\end{equation}
The trajectories in the $\left(  y,V\right)  $-plane associated to the
solutions of (\ref{S4E10}), (\ref{S4E11}) are contained in the level sets:
\begin{equation}
H=h. \label{S5E4}%
\end{equation}

We will also need the self-similar formulation of the integrated form of the
equation (\ref{S3E7}). In this case the function $\zeta$ in (\ref{S3E6}) has
the functional dependence:
\[
\zeta\left(  r,v,t\right)  =\left(  -t\right)  ^{2}\Theta\left(  y,V\right)
.
\]
Notice that:
\begin{equation}
\Theta\left(  y,V\right)  =\int_{0}^{\infty}G\Phi d\Phi. \label{S5E4a}%
\end{equation}
The function $\Theta$ satisfies:%

\begin{align}
&  y\Theta_{y}-V\Theta_{V}-2\Theta+e^{U-\Lambda}\frac{V}{\hat{E}}\Theta
_{y}\nonumber\\
&  -\left(  y\Lambda_{y}V+e^{U-\Lambda}U_{y}\hat{E}-e^{U-\Lambda}\frac
{1}{y^{3}\hat{E}}\right)  \Theta_{V}\nonumber\\
&  =0 \label{S5E8}%
\end{align}
and:
\begin{align}
\tilde{\rho}  &  =\frac{\pi}{y^{2}}\int_{-\infty}^{\infty}\hat{E}\Theta
dV,\label{S5E6}\\
\tilde{p}  &  =\frac{\pi}{y^{2}}\int_{-\infty}^{\infty}\frac{V^{2}}{\hat{E}%
}\Theta dV. \label{S5E7}%
\end{align}
The characteristic curves associated to (\ref{S5E8}) are (\ref{S4E10}),
(\ref{S4E11}) and:
\begin{equation}
\frac{d\Theta}{d\sigma}=2\Theta. \label{S5E9}%
\end{equation}

\bigskip

\section{\label{SectionSSS}SINGULAR SELF-SIMILAR SOLUTIONS: \protect\linebreak
GENERAL PROPERTIES.}

\bigskip

The main goal of this paper is to construct a family of distributional
solutions of (\ref{S4E9})-(\ref{S4E8}) for which $G=G\left(  y,V,\Phi\right)
$ is a measure supported on some surfaces in the three-dimensional space with
coordinates $\left(  y,V,\Phi\right)  .$ In this section we will describe in a
heuristic manner the argument yielding the construction of such solutions. The
arguments will be made rigorous in the rest of the paper. The key idea behind
the argument is that the problem can be transformed into a system of ordinary
differential equations for the particular class of solutions described in this section.

Taking into account that the singularities of the distribution $G$ might be
expected to be propagated by characteristics it is natural to look for
solutions of (\ref{S4E9})-(\ref{S4E8}) of the form:
\begin{equation}
G\left(  y,V,\Phi\right)  =A\left(  y,V,\Phi\right)  \delta\left(  H\left(
y,V\right)  -h\right)  \label{S6E4}%
\end{equation}
satisfying (\ref{S4E9}) in the sense of distributions. Let us assume that
$A,\;H$ have the differentiability properties required for all the following
formal computations. Plugging (\ref{S6E4}) into (\ref{S4E9}) we obtain:
\begin{align*}
&  \left(  a\left(  y,V\right)  A_{y}+b\left(  y,V\right)  A_{V}+\Phi A_{\Phi
}\right)  \delta\left(  H-h\right) \\
&  +A\left(  a\left(  y,V\right)  H_{y}+b\left(  y,V\right)  H_{V}\right)
\delta^{\prime}\left(  H-h\right) \\
&  =0
\end{align*}
where:
\begin{align}
a\left(  y,V\right)   &  \equiv y+e^{U-\Lambda}\frac{V}{\hat{E}}=e^{\Lambda
}\frac{\partial H}{\partial V},\label{S6E4a1}\\
b\left(  y,V\right)   &  \equiv-V-\left(  y\Lambda_{y}V+e^{U-\Lambda}U_{y}%
\hat{E}-e^{U-\Lambda}\frac{1}{y^{3}\hat{E}}\right)  =-e^{\Lambda}%
\frac{\partial H}{\partial y}. \label{S6E4a2}%
\end{align}
Notice that $a\left(  y,V\right)  H_{y}+b\left(  y,V\right)  H_{V}=0.$ Then:
\[
\left(  a\left(  y,V\right)  A_{y}+b\left(  y,V\right)  A_{V}+\Phi A_{\Phi
}\right)  \delta\left(  H-h\right)  =0.
\]
This equation is satisfied if:
\begin{equation}
a\left(  y,V\right)  A_{y}+b\left(  y,V\right)  A_{V}+\Phi A_{\Phi}=0
\label{S6E4a}%
\end{equation}
on the surface $\left\{  H=h\right\}  \times\mathbb{R}^{+}.$ Let us assume
that the curve $\left\{  H=h\right\}  $ can be parametrized, at least locally,
using a parameter $\sigma$ satisfying:
\begin{align}
y  &  =y\left(  \sigma\right)  \;\;,\;\;V=V\left(  \sigma\right)  ,\nonumber\\
\frac{dy\left(  \sigma\right)  }{d\sigma}  &  =a\left(  y\left(
\sigma\right)  ,V\left(  \sigma\right)  \right)  \;\;,\;\;\frac{dV\left(
\sigma\right)  }{d\sigma}=b\left(  y\left(  \sigma\right)  ,V\left(
\sigma\right)  \right)  . \label{S6E4b}%
\end{align}
Then the function $A$ can be written on the surface $\left\{  H=h\right\}
\times\mathbb{R}^{+}$ as a function of the variables $\left(  \sigma
,\Phi\right)  .$ We can write:
\begin{equation}
A\left(  y\left(  \sigma\right)  ,V\left(  \sigma\right)  ,\Phi\right)
=\bar{A}\left(  \sigma,\Phi\right)  \;\;\text{for }\left(  y,V,\sigma\right)
\in\left\{  H=h\right\}  \times\mathbb{R}^{+} \label{S6E4bb}%
\end{equation}
and using (\ref{S6E4b}) we can rewrite (\ref{S6E4a}) as:
\begin{equation}
\bar{A}_{\sigma}+\Phi\bar{A}_{\Phi}=0. \label{S6E4c}%
\end{equation}

Since the curves $\left\{  H=h\right\}  $ can be determined in terms of
$\Theta$ alone it is convenient to compute this distribution explicitly. If
$G$ has the form (\ref{S6E4}) the distribution $\Theta$ defined in
(\ref{S5E4a}) is given by:
\begin{equation}
\Theta\left(  y,V\right)  =\beta\delta\left(  H-h\right)  \label{S6E4c1}%
\end{equation}
where:
\[
\beta=\int_{0}^{\infty}A\Phi d\Phi.
\]
Since $A$ is given by (\ref{S6E4bb}) it follows that:
\begin{equation}
\beta=\beta\left(  \sigma\right)  =\int_{0}^{\infty}\bar{A}\left(  \sigma
,\Phi\right)  \Phi d\Phi\;\;\text{for\ \ }\left(  y,V\right)  \in\left\{
H=h\right\}  . \label{S6E4e}%
\end{equation}
We can compute $\beta\left(  \sigma\right)  $ along the curve $\left\{
H=h\right\}  .$ To this end we multiply (\ref{S6E4c}) by $\Phi$ and integrate
in the $\Phi$ variable in the interval $\left[  0,\infty\right)  .$ Then:
\[
\beta_{\sigma}=2\beta.
\]
The function $\beta$ then takes the form:
\begin{equation}
\beta\left(  \sigma\right)  =\beta_{0}e^{2\sigma} \label{S6E4d}%
\end{equation}
for some $\beta_{0}\geq0.$

In the rest of the paper we prove that there exist functions $\bar{A}\left(
\sigma,\Phi\right)  $ as in (\ref{S6E4bb}) and curves $\left\{  H=h\right\}  $
with $\Lambda,\;U$ solving (\ref{S4E3})-(\ref{S4E6}) and $\tilde{\rho
},\;\tilde{p}$ as in (\ref{S4E7}), (\ref{S4E8}) such that (\ref{S6E4}) solves
(\ref{S4E9}) in the sense of distributions.

\bigskip

\section{SINGULAR SELF-SIMILAR SOLUTIONS:\protect\linebreak DESCRIBING THEIR
SUPPORT.\label{SelfSimD}}

\bigskip

In this section we describe in a precise manner the form of the curved surface
containing the support of the distribution $G$ for the self-similar solutions
constructed in this paper. Such a surface is contained in the surface
$S=\gamma\times\mathbb{R}^{+},$ where $\gamma\subset\left\{  \left(
y,V\right)  :y>0\;,\;V\in\mathbb{R}\right\}  $ is an unbounded curve, at a
strictly positive distance from the line $\left\{  y=0\right\}  \equiv\left\{
\left(  y,V\right)  :y=0\;,\;V\in\mathbb{R}\right\}  $ with a discontinuity in
its curvature at the point $\left(  y_{0},V_{0}\right)  \in\gamma$ placed at
the minimum distance from the line $\left\{  y=0\right\}  .$ In order to avoid
such irregular curves it is more convenient to assume that the curve $\gamma$
is the union of two analytic curves $\gamma_{1}$ and $\gamma_{2}$ that can be
parametrized in the form:
\begin{equation}
\gamma_{i}=\left\{  \left(  y,V\right)  :y_{0}<y<\infty\;,\;V=V_{i}\left(
y\right)  \right\}  \;\;,\;\;i=1,2 \label{S7E0}%
\end{equation}
where the functions $V_{i}\left(  y\right)  $ are analytic and satisfy:%

\begin{align}
\lim_{y\rightarrow y_{0}^{+}}V_{1}\left(  y\right)   &  =\lim_{y\rightarrow
y_{0}^{+}}V_{2}\left(  y\right)  =V_{0}=-\frac{1}{\sqrt{1-y_{0}^{2}}%
},\label{S7E1}\\
V_{1}\left(  y\right)   &  <V_{2}\left(  y\right)  \;\;\text{for\ \ }%
y_{0}<y<\infty\label{S7E2}%
\end{align}
for some $y_{0}\in\left(  0,1\right)  .$ Since the curves $\gamma_{i}$ are
contained in the curve $\left\{  H=h\right\}  $ it follows that the functions
$V_{i}\left(  y\right)  $ are the two roots of the equation:
\begin{equation}
\frac{e^{U}}{y}\sqrt{V^{2}y^{2}+1}+yVe^{\Lambda}=h \label{S7E2a}%
\end{equation}
assuming that such roots exist. Then:%
\begin{align}
V_{1}\left(  y\right)   &  =\frac{1}{\left(  e^{2U}-y^{2}e^{2\Lambda}\right)
}\left[  -ye^{\Lambda}h-\sqrt{\left(  ye^{\Lambda}h\right)  ^{2}-\left(
e^{2U}-y^{2}e^{2\Lambda}\right)  \left(  \frac{e^{2U}}{y^{2}}-h^{2}\right)
}\right]  ,\label{S7E3}\\
V_{2}\left(  y\right)   &  =\frac{1}{\left(  e^{2U}-y^{2}e^{2\Lambda}\right)
}\left[  -ye^{\Lambda}h+\sqrt{\left(  ye^{\Lambda}h\right)  ^{2}-\left(
e^{2U}-y^{2}e^{2\Lambda}\right)  \left(  \frac{e^{2U}}{y^{2}}-h^{2}\right)
}\right]  . \label{S7E4}%
\end{align}
Notice that for such solutions the support of $G$ in (\ref{S6E4}) is contained
in the half-plane $\left\{  y\geq y_{0}\right\}  .$ Therefore, $\rho\left(
y\right)  =p\left(  y\right)  =0$ for $y<y_{0}.$ Then (\ref{S4E3}%
)-(\ref{S4E6}) imply $U\left(  y\right)  =\Lambda\left(  y\right)  =0$ for
$y<y_{0}. $

Under suitable regularity assumptions for the curves $\gamma_{i}$ near the
point $\left(  y_{0},V_{0}\right)  $ that will be made precise below the
functions $U$ and $\Lambda$ are continuous at the point $y=y_{0}$. In such a
case (\ref{S7E2a}) implies:
\begin{equation}
h=\frac{\sqrt{V_{0}^{2}y_{0}^{2}+1}}{y_{0}}+y_{0}V_{0}=\frac{\sqrt{1-y_{0}%
^{2}}}{y_{0}}. \label{S7E5}%
\end{equation}
We will prove later that it is possible to construct the desired curves
$\gamma_{i}\;,\;i=1,2$, defined by means of (\ref{S7E0}) with the property
that the following limits exist:
\begin{equation}
\lim_{y\rightarrow y_{0}^{+}}\frac{V_{i}\left(  y\right)  -V_{0}}%
{\sqrt{y-y_{0}}}=K_{i}\;\;,\;\;K_{i}\in\mathbb{R}\text{\ \ ,\ \ }%
i=1,2\;,\;K_{1}<K_{2}. \label{S7E6}%
\end{equation}
Moreover, the quotients of the functions $\Lambda$ and $U$ by $\sqrt{y-y_{0}}$
also tend to finite limits. Let:
\begin{align}
\lim_{y\rightarrow y_{0}^{+}}\frac{\Lambda\left(  y\right)  }{\sqrt{y-y_{0}}}
&  =\theta_{1}\in\mathbb{R},\label{S7E6a}\\
\lim_{y\rightarrow y_{0}^{+}}\frac{U\left(  y\right)  }{\sqrt{y-y_{0}}}  &
=\theta_{2}\in\mathbb{R}. \label{S7E6b}%
\end{align}
We parametrize the curve $\gamma=\left\{  H=h\right\}  $ as in the previous
section using a parameter $\sigma.$ We will denote a parameter of this kind on
the curves $\gamma_{1},\;\gamma_{2}$ by $\sigma_{1},\;\sigma_{2}$
respectively. Due to (\ref{S6E4a1}), (\ref{S6E4b}), (\ref{S7E0}) it follows
that:
\begin{equation}
\frac{d\sigma_{i}}{dy}=\frac{1}{a\left(  y,V_{i}\left(  y\right)  \right)
}=\frac{1}{y+e^{U-\Lambda}\frac{V_{i}\left(  y\right)  y}{\sqrt{\left(
V_{i}\left(  y\right)  \right)  ^{2}y^{2}+1}}}\;\;,\;\;i=1,2. \label{S7E7}%
\end{equation}
We will normalize the parameters $\sigma_{i}=\sigma_{i}\left(  y\right)  $ by
means of the condition:
\begin{equation}
\sigma_{i}\left(  y_{0}\right)  =0\;\;,\;\;i=1,2. \label{S7E8}%
\end{equation}

Finally we remark that in order to obtain the functions $U$ and $\Lambda$ we
need to prescribe the distribution $\Theta$ defined by (\ref{S5E4a}). Using
(\ref{S6E4c1}), (\ref{S6E4d}) it then follows that:
\begin{align}
\Theta\left(  y,V\right)   &  =\frac{\beta_{0}\chi_{\left\{  y>y_{0}\right\}
}e^{2\sigma_{1}\left(  y\right)  }}{\left|  \frac{\partial H}{\partial
V}\left(  y,V_{1}\left(  y\right)  \right)  \right|  }\delta\left(
V-V_{1}\left(  y\right)  \right) \nonumber\\
&  +\frac{\beta_{0}\chi_{\left\{  y>y_{0}\right\}  }e^{2\sigma_{2}\left(
y\right)  }}{\left|  \frac{\partial H}{\partial V}\left(  y,V_{2}\left(
y\right)  \right)  \right|  }\delta\left(  V-V_{2}\left(  y\right)  \right)
\label{S7E9}%
\end{align}
where $\chi_{\left\{  y>y_{0}\right\}  }$ is the characteristic function of
the half-plane $\left\{  y>y_{0}\right\}  .$ Using (\ref{S5E6}), (\ref{S5E7})
it follows that:
\begin{align}
\tilde{\rho}\left(  y\right)   &  =\frac{\pi\beta_{0}\chi_{\left\{
y>y_{0}\right\}  }}{y^{3}}\left[  \frac{e^{2\sigma_{1}\left(  y\right)  }%
}{\left|  \frac{\partial H}{\partial V}\left(  y,V_{1}\left(  y\right)
\right)  \right|  }\sqrt{\left(  V_{1}\left(  y\right)  \right)  ^{2}y^{2}%
+1}\right. \nonumber\\
&  \left.  +\frac{e^{2\sigma_{2}\left(  y\right)  }}{\left|  \frac{\partial
H}{\partial V}\left(  y,V_{2}\left(  y\right)  \right)  \right|  }%
\sqrt{\left(  V_{2}\left(  y\right)  \right)  ^{2}y^{2}+1}\right]
,\label{S7E10a}\\
\tilde{p}\left(  y\right)   &  =\frac{\pi\beta_{0}\chi_{\left\{
y>y_{0}\right\}  }}{y}\left[  \frac{e^{2\sigma_{1}\left(  y\right)  }}{\left|
\frac{\partial H}{\partial V}\left(  y,V_{1}\left(  y\right)  \right)
\right|  }\frac{\left(  V_{1}\left(  y\right)  \right)  ^{2}}{\sqrt{\left(
V_{1}\left(  y\right)  \right)  ^{2}y^{2}+1}}\right. \nonumber\\
&  \left.  +\frac{e^{2\sigma_{2}\left(  y\right)  }}{\left|  \frac{\partial
H}{\partial V}\left(  y,V_{2}\left(  y\right)  \right)  \right|  }%
\frac{\left(  V_{2}\left(  y\right)  \right)  ^{2}}{\sqrt{\left(  V_{2}\left(
y\right)  \right)  ^{2}y^{2}+1}}\right]  \label{S7E10b}%
\end{align}
and the functions $U$ and $\Lambda$ can then be obtained using the equations
(\ref{S4E3}), (\ref{S4E4}).

Due to the dust-like character of the solutions considered in this paper, they
exhibit a singular behaviour for $\tilde{\rho}$ and $\tilde{p}$ at the radius
$y=y_{0}.$ This singularity is due to the fact that at this point the radial
velocity of the particles, in self-similar variables, vanishes. However, since
the motion of the trajectories after they reach the singularity continues in a
smooth way, and since $\tilde{\rho}$ and $\tilde{p}$ are integrable near this
radius, this singularity can be expected to disappear if the dust-like
assumption is relaxed and some thickness is given to the support of the
distribution function in the phase space.

The main result of this paper is the following:

\bigskip

\begin{theorem}
\label{Th1} There exists $\varepsilon_{0}>0$ small such that, for any
$y_{0}\in\left(  0,\varepsilon_{0}\right)  $ there exist a value of $\beta
_{0}>0$ and two curves $\gamma_{1},\;\gamma_{2}$ that can be parametrized as
in (\ref{S7E0}) with the functions $V_{1}\left(  y\right)  ,\;V_{2}\left(
y\right)  $ as in (\ref{S7E3}), (\ref{S7E4}) satisfying (\ref{S7E1}),
(\ref{S7E2}), (\ref{S7E6}), the functions $U,\;\Lambda$ satisfying
(\ref{S4E3}), (\ref{S4E4}) and (\ref{S7E6a}), (\ref{S7E6b}) with $\tilde{\rho
},\;\tilde{p}$ as in (\ref{S7E10a}), (\ref{S7E10b}) and $\sigma_{1}%
,\;\sigma_{2}$ solving (\ref{S7E7}), (\ref{S7E8}).
\end{theorem}

\bigskip

Using Theorem \ref{Th1} it is possible to obtain distributional solutions of
the problem (\ref{S4E9})-(\ref{S4E8}). In order to make the definition of the
distribution $G$ in (\ref{S6E4}) precise we use (\ref{S6E4bb}), (\ref{S6E4c}).
Let us prescribe a smooth function $\bar{A}_{0}\left(  \Phi\right)  $ in
$\Phi\in\left(  0,\infty\right)  .$ Taking into account (\ref{S6E4c}) we can
then define:
\[
\bar{A}\left(  \sigma,\Phi\right)  =\bar{A}_{0}\left(  e^{-\sigma}\Phi\right)
.
\]
Using the structure of the curves $\gamma_{1},\;\gamma_{2}$ it would then
follow that the distribution $G$ in (\ref{S6E4}) would be given by:
\begin{align}
G\left(  y,V,\Phi\right)   &  =\frac{\bar{A}_{0}\left(  e^{-\sigma_{1}\left(
y\right)  }\Phi\right)  \chi_{\left\{  y>y_{0}\right\}  }}{\left|
\frac{\partial H}{\partial V}\left(  y,V_{1}\left(  y\right)  \right)
\right|  }\delta\left(  V-V_{1}\left(  y\right)  \right) \nonumber\\
&  +\frac{\bar{A}_{0}\left(  e^{-\sigma_{2}\left(  y\right)  }\Phi\right)
\chi_{\left\{  y>y_{0}\right\}  }}{\left|  \frac{\partial H}{\partial
V}\left(  y,V_{2}\left(  y\right)  \right)  \right|  }\delta\left(
V-V_{2}\left(  y\right)  \right)  . \label{S7E11}%
\end{align}
We then have the following result:

\bigskip

\begin{theorem}
\label{Th2} Suppose that the function $\bar{A}_{0}\left(  \cdot\right)  \in
C_{0}^{1}\left(  0,\infty\right)  $ satisfies
\begin{equation}
\int_{0}^{\infty}\bar{A}_{0}\left(  \Phi\right)  \Phi d\Phi=\beta_{0}.
\label{S7E11a}%
\end{equation}
Let us define a Radon measure $G\in\mathcal{M}\left(  \mathbb{R}^{+}%
\times\mathbb{R}\times\mathbb{R}^{+}\right)  $ by means of (\ref{S7E11}) with
the functions $V_{1}\left(  \cdot\right)  ,\;V_{2}\left(  \cdot\right)
,\;\sigma_{1}\left(  \cdot\right)  ,\;\sigma_{2}\left(  \cdot\right)  \;$
as in Theorem \ref{Th1}. Then the functions $\tilde{\rho},\;\tilde{p}$ defined
(\ref{S4E7}), (\ref{S4E8}) belong to the spaces $L_{loc}^{p}\left(
0,\infty\right)  $ for $1\leq p<2.$ The functions $\Lambda,\;U$ defined by
means of (\ref{S4E3})-(\ref{S4E6}) belong to $W_{loc}^{1,p}\left(
0,\infty\right)  $ for $1\leq p<2.$ The measure $G$ satisfies (\ref{S4E9}) in
the sense of distributions.
\end{theorem}

\begin{remark}
The space $C_{0}^{1}\left(  0,\infty\right)  $ is the space of compactly
supported continuously differentiable functions and the space $\mathcal{M}%
\left(  \mathbb{R}^{+}\times\mathbb{R}\times\mathbb{R}^{+}\right)  $ is the
space of Radon measures on $\mathbb{R}^{+}\times\mathbb{R}\times\mathbb{R}%
^{+}$. It is not necessary to require $A_{0}\left(  \cdot\right)  $ to be
compactly supported. Actually this condition could be replaced by assumptions
of fast enough decay near the origin and infinity.
\end{remark}

\begin{remark}
It is worth noticing that the functions $\tilde{\rho},\;\tilde{p}$ associated
to the distribution $G$ have an integrable singularity as $y\rightarrow
y_{0}^{+}.$
\end{remark}

In the rest of this section we will prove Theorem \ref{Th2}. Theorem \ref{Th1}
will be proved in the remaining sections of the paper using a shooting
argument and refined asymptotics of the solutions for $y_{0}$ small. The
following auxiliary result will be used in the proof of Theorem \ref{Th2} and
it will be proved in Section \ref{ProT1}. We remark that Theorem \ref{Th2}
will not be used in either the proof of Theorem \ref{Th1} or that of
Proposition \ref{Prop1} below.

\bigskip

\begin{proposition}
\label{Prop1}The curves $\gamma_{1},\;\gamma_{2}$ whose existence has been
proved in Theorem \ref{Th1} satisfy the following conditions:
\begin{equation}
\lim_{y\rightarrow y_{0}^{+}}\frac{\frac{\partial H}{\partial V}\left(
y,V_{1}\left(  y\right)  \right)  }{\sqrt{y-y_{0}}}=L_{1}\;\;,\;\;\lim
_{y\rightarrow y_{0}^{+}}\frac{\frac{\partial H}{\partial V}\left(
y,V_{2}\left(  y\right)  \right)  }{\sqrt{y-y_{0}}}=L_{2} \label{S8E1}%
\end{equation}
for some constants $L_{1}<L_{2}$.
\end{proposition}

\begin{proof}
[Proof of Theorem \ref{Th2}]Using (\ref{S5E4a}), (\ref{S7E11}) and
(\ref{S7E11a}) we obtain:
\begin{align}
\Theta\left(  y,V\right)   &  =\int_{0}^{\infty}G\Phi d\Phi=\frac{\beta
_{0}e^{2\sigma_{1}\left(  y\right)  }\chi_{\left\{  y>y_{0}\right\}  }%
}{\left|  \frac{\partial H}{\partial V}\left(  y,V_{1}\left(  y\right)
\right)  \right|  }\delta\left(  V-V_{1}\left(  y\right)  \right) \nonumber\\
&  +\frac{\beta_{0}e^{2\sigma_{2}\left(  y\right)  }\chi_{\left\{
y>y_{0}\right\}  }}{\left|  \frac{\partial H}{\partial V}\left(
y,V_{2}\left(  y\right)  \right)  \right|  }\delta\left(  V-V_{2}\left(
y\right)  \right)  .
\end{align}
We can then compute $\tilde{\rho},\;\tilde{p}$ using (\ref{S5E6}),
(\ref{S5E7}):
\begin{align}
\tilde{\rho}\left(  y\right)   &  =\frac{\pi\beta_{0}\chi_{\left\{
y>y_{0}\right\}  }}{y^{3}}\left[  \frac{e^{2\sigma_{1}\left(  y\right)  }%
\sqrt{1+y^{2}\left(  V_{1}\left(  y\right)  \right)  ^{2}}}{\left|
\frac{\partial H}{\partial V}\left(  y,V_{1}\left(  y\right)  \right)
\right|  }+\frac{e^{2\sigma_{2}\left(  y\right)  }\sqrt{1+y^{2}\left(
V_{2}\left(  y\right)  \right)  ^{2}}}{\left|  \frac{\partial H}{\partial
V}\left(  y,V_{2}\left(  y\right)  \right)  \right|  }\right]  ,\label{Tr1}\\
\tilde{p}\left(  y\right)   &  =\frac{\pi\beta_{0}\chi_{\left\{
y>y_{0}\right\}  }}{y}\left[  \frac{e^{2\sigma_{1}\left(  y\right)  }\left(
V_{1}\left(  y\right)  \right)  ^{2}}{\left|  \frac{\partial H}{\partial
V}\left(  y,V_{1}\left(  y\right)  \right)  \right|  \sqrt{1+y^{2}\left(
V_{1}\left(  y\right)  \right)  ^{2}}}\right. \nonumber\\
&  \left.  +\frac{e^{2\sigma_{2}\left(  y\right)  }\left(  V_{2}\left(
y\right)  \right)  ^{2}}{\left|  \frac{\partial H}{\partial V}\left(
y,V_{2}\left(  y\right)  \right)  \right|  \sqrt{1+y^{2}\left(  V_{2}\left(
y\right)  \right)  ^{2}}}\right]  . \label{Tp1}%
\end{align}
Using (\ref{S7E6})-(\ref{S7E6b}), (\ref{S8E1}), we obtain:
\begin{equation}
\left\vert \tilde{\rho}\left(  y\right)  \right\vert +\left\vert \tilde
{p}\left(  y\right)  \right\vert \leq\frac{C\chi_{\left\{  y>y_{0}\right\}  }%
}{\sqrt{y-y_{0}}} \label{S8E2}%
\end{equation}
whence the estimate $\tilde{\rho},\;\tilde{p}\in L_{loc}^{p}\left(
0,\infty\right)  ,\;1\leq p<2$, in the theorem follows. On the other hand
(\ref{S4E3})-(\ref{S4E6}) imply:
\begin{align}
\Lambda &  =-\frac{1}{2}\log\left(  1-\frac{8\pi}{y}\int_{y_{0}}^{y}\xi
^{2}\tilde{\rho}\left(  \xi\right)  d\xi\right)  ,\label{S8E3}\\
U  &  =\int_{y_{0}}^{y}\frac{\left[  \left(  8\pi\xi^{2}\tilde{p}\left(
\xi\right)  +1\right)  e^{2\Lambda\left(  \xi\right)  }-1\right]  }{2\xi}d\xi.
\label{S8E4}%
\end{align}
Due to Theorem \ref{Th1} the functions $\Lambda,\;U$ are bounded for any
finite value $y>0.$ On the other hand, (\ref{S8E3}), (\ref{S8E4}) imply
$\Lambda,\;U\in W_{loc}^{1,p}\left(  0,\infty\right)  ,\;1\leq p<2.$

In order to conclude the proof of Theorem \ref{Th2} it only remains to prove
that $G$ solves (\ref{S4E9}) in the sense of distributions. This is equivalent
to showing that:
\begin{align}
&  \int_{\mathbb{R}^{+}\times\mathbb{R}\times\mathbb{R}^{+}}\left[  -\left(
y\varphi\right)  _{y}+\left(  V\varphi\right)  _{V}-\left(  \Phi
\varphi\right)  _{\Phi}-\left(  e^{U-\Lambda}\frac{V}{\hat{E}}\varphi\right)
_{y}\right. \nonumber\\
&  \left.  +\left(  \left(  y\Lambda_{y}V+e^{U-\Lambda}U_{y}\hat
{E}-e^{U-\Lambda}\frac{1}{y^{3}\hat{E}}\right)  \varphi\right)  _{V}\right]
GdydVd\Phi\nonumber\\
&  =0 \label{S8E5}%
\end{align}
for any $\varphi=\varphi\left(  y,V,\Phi\right)  \in C_{0}^{\infty}\left(
\mathbb{R}^{+}\times\mathbb{R}\times\mathbb{R}^{+}\right)  .$ Using
(\ref{S7E11}) we can rewrite (\ref{S8E5}) as:
\begin{align}
&  J\equiv\sum_{i=1}^{2}\int_{y_{0}}^{\infty}\int_{0}^{\infty}\left[  -\left(
y\varphi\right)  _{y}+\left(  V\varphi\right)  _{V}-\left(  \Phi
\varphi\right)  _{\Phi}-\left(  e^{U-\Lambda}\frac{V}{\hat{E}}\varphi\right)
_{y}\right.  +\label{S8E5a}\\
&  \left.  \left.  \left(  \left(  y\Lambda_{y}V+e^{U-\Lambda}U_{y}\hat
{E}-e^{U-\Lambda}\frac{1}{y^{3}\hat{E}}\right)  \varphi\right)  _{V}\right]
\right|  _{\left(  y,V_{i}\left(  y\right)  ,\Phi\right)  }\frac{\bar{A}%
_{0}\left(  e^{-\sigma_{i}\left(  y\right)  }\Phi\right)  }{\left|
\frac{\partial H}{\partial V}\left(  y,V_{i}\left(  y\right)  \right)
\right|  }d\Phi dy\nonumber\\
&  =0\nonumber
\end{align}
and making the change of variables $e^{-\sigma_{i}\left(  y\right)  }%
\Phi\rightarrow\Phi$ we can transform $J$\ into:
\begin{align}
&  J\equiv\sum_{i=1}^{2}\int_{y_{0}}^{\infty}\int_{0}^{\infty}\left[  -\left(
y\varphi\right)  _{y}+\left(  V\varphi\right)  _{V}-\left(  \Phi
\varphi\right)  _{\Phi}-\left(  e^{U-\Lambda}\frac{V}{\hat{E}}\varphi\right)
_{y}\right.  +\nonumber\\
&  \left.  \left.  \left(  \left(  y\Lambda_{y}V+e^{U-\Lambda}U_{y}\hat
{E}-e^{U-\Lambda}\frac{1}{y^{3}\hat{E}}\right)  \varphi\right)  _{V}\right]
\right|  _{\left(  y,V_{i}\left(  y\right)  ,\Phi e^{\sigma_{i}\left(
y\right)  }\right)  }\frac{\bar{A}_{0}\left(  \Phi\right)  e^{\sigma
_{i}\left(  y\right)  }}{\left|  \frac{\partial H}{\partial V}\left(
y,V_{i}\left(  y\right)  \right)  \right|  }dyd\Phi\label{S8E6}%
\end{align}
Notice that we can write:
\begin{align*}
F  &  \equiv-\left(  y\varphi\right)  _{y}+\left(  V\varphi\right)
_{V}-\left(  \Phi\varphi\right)  _{\Phi}-\left(  e^{U-\Lambda}\frac{V}{\hat
{E}}\varphi\right)  _{y}\\
&  +\left(  \left(  y\Lambda_{y}V+e^{U-\Lambda}U_{y}\hat{E}-e^{U-\Lambda}%
\frac{1}{y^{3}\hat{E}}\right)  \varphi\right)  _{V}\\
&  =-y\varphi_{y}+V\varphi_{V}-\left(  \Phi\varphi\right)  _{\Phi}-\left(
e^{U-\Lambda}\frac{Vy}{\sqrt{1+V^{2}y^{2}}}\varphi\right)  _{y}\\
&  +\left(  \left(  y\Lambda_{y}V+\frac{e^{U-\Lambda}U_{y}}{y}\sqrt
{1+V^{2}y^{2}}-e^{U-\Lambda}\frac{1}{y^{2}\sqrt{1+V^{2}y^{2}}}\right)
\varphi\right)  _{V}%
\end{align*}
and, using Leibniz's rule:
\begin{align*}
F  &  =-y\Lambda_{y}\varphi-ye^{\Lambda}\left(  e^{-\Lambda}\varphi\right)
_{y}+V\varphi_{V}-\Phi\varphi_{\Phi}-\varphi-U_{y}e^{U-\Lambda}\frac{Vy}%
{\sqrt{1+V^{2}y^{2}}}\varphi\\
&  -e^{U-\Lambda}\frac{V}{\sqrt{1+V^{2}y^{2}}}\varphi+e^{U-\Lambda}\frac
{V^{3}y^{2}}{\left(  1+V^{2}y^{2}\right)  ^{\frac{3}{2}}}\varphi-e^{U}%
\frac{Vy}{\sqrt{1+V^{2}y^{2}}}\left(  e^{-\Lambda}\varphi\right)  _{y}\\
&  +\left(  y\Lambda_{y}+U_{y}e^{U-\Lambda}\frac{Vy}{\sqrt{1+V^{2}y^{2}}%
}+e^{U-\Lambda}\frac{V}{\left(  1+V^{2}y^{2}\right)  ^{\frac{3}{2}}}\right)
\varphi\\
&  +\left(  y\Lambda_{y}V+\frac{e^{U-\Lambda}U_{y}}{y}\sqrt{1+V^{2}y^{2}%
}-e^{U-\Lambda}\frac{1}{y^{2}\sqrt{1+V^{2}y^{2}}}\right)  \varphi_{V}.
\end{align*}
After some cancellations:
\begin{align}
&  F\left(  y,V,\Phi\right)  =-\left(  ye^{\Lambda}+e^{U}\frac{Vy}%
{\sqrt{1+V^{2}y^{2}}}\right)  \left(  e^{-\Lambda}\varphi\right)  _{y}-\Phi
e^{\Lambda}\left(  e^{-\Lambda}\varphi\right)  _{\Phi}-\varphi\nonumber\\
&  +\left(  y\Lambda_{y}Ve^{\Lambda}+Ve^{\Lambda}+\frac{e^{U}U_{y}}{y}%
\sqrt{1+V^{2}y^{2}}-e^{U}\frac{1}{y^{2}\sqrt{1+V^{2}y^{2}}}\right)  \left(
e^{-\Lambda}\varphi\right)  _{V}. \label{S8E7}%
\end{align}
Then (\ref{S8E6}) can be rewritten as:
\[
\sum_{i=1}^{2}\int_{y_{0}}^{\infty}\int_{0}^{\infty}F\left(  y,V_{i}\left(
y\right)  ,\Phi e^{\sigma_{i}\left(  y\right)  }\right)  \frac{\bar{A}%
_{0}\left(  \Phi\right)  e^{\sigma_{i}\left(  y\right)  }}{\left|
\frac{\partial H}{\partial V}\left(  y,V_{i}\left(  y\right)  \right)
\right|  }d\Phi dy=0.
\]
Due to Proposition \ref{Prop1} as well as the fact that the curves $\gamma
_{1},\;\gamma_{2}$ are globally defined it follows that:
\[
\left|  \frac{\partial H}{\partial V}\left(  y,V_{i}\left(  y\right)  \right)
\right|  =\left(  -1\right)  ^{i-1}\frac{\partial H}{\partial V}\left(
y,V_{i}\left(  y\right)  \right)  .
\]
Then:
\begin{equation}
J=\sum_{i=1}^{2}\left(  -1\right)  ^{i-1}\int_{y_{0}}^{\infty}\int_{0}%
^{\infty}F\left(  y,V_{i}\left(  y\right)  ,\Phi e^{\sigma_{i}\left(
y\right)  }\right)  \frac{\bar{A}_{0}\left(  \Phi\right)  e^{\sigma_{i}\left(
y\right)  }}{\frac{\partial H}{\partial V}\left(  y,V_{i}\left(  y\right)
\right)  }d\Phi dy. \label{S8E8}%
\end{equation}
Using (\ref{S5E3}) and (\ref{S8E7}) we obtain:
\begin{align*}
&  \frac{F\left(  y,V,\Phi\right)  }{ye^{\Lambda}+e^{U}\frac{Vy}{\sqrt
{V^{2}y^{2}+1}}} =-\left(  e^{-\Lambda}\varphi\right)  _{y}-\frac{\Phi
}{y+e^{U-\Lambda}\frac{Vy}{\sqrt{V^{2}y^{2}+1}}}\left(  e^{-\Lambda}%
\varphi\right)  _{\Phi}\\
&  -\frac{1}{y+e^{U-\Lambda}\frac{Vy}{\sqrt{V^{2}y^{2}+1}}}\left(
e^{-\Lambda}\varphi\right) \\
&  +\frac{1}{y+e^{U-\Lambda}\frac{Vy}{\sqrt{V^{2}y^{2}+1}}}\left(
y\Lambda_{y}V+V+\frac{e^{U-\Lambda}U_{y}}{y}\sqrt{1+V^{2}y^{2}}\right)
\left(  e^{-\Lambda}\varphi\right)  _{V}\\
&  -\frac{e^{U-\Lambda}}{y+e^{U-\Lambda}\frac{Vy}{\sqrt{V^{2}y^{2}+1}}}%
\frac{1}{y^{2}\sqrt{1+V^{2}y^{2}}}\left(  e^{-\Lambda}\varphi\right)  _{V}.
\end{align*}
Equations (\ref{S6E4a2}), (\ref{S6E4a}) and (\ref{S7E7}) give:
\begin{align}
&  \frac{dV_{i}}{dy}\left(  y\right)  =-\frac{1}{y+e^{U-\Lambda}\frac
{Vy}{\sqrt{V^{2}y^{2}+1}}}\left(  y\Lambda_{y}V+V+\frac{e^{U-\Lambda}U_{y}}%
{y}\sqrt{1+V^{2}y^{2}}\right. \nonumber\\
&  \left.  -e^{U-\Lambda}\frac{1}{y^{2}\sqrt{1+V^{2}y^{2}}}\right)  .
\end{align}
Therefore
\begin{align*}
&  F\left(  y,V_{i}\left(  y\right)  ,\Phi e^{\sigma_{i}\left(  y\right)
}\right)  \frac{e^{\sigma_{i}\left(  y\right)  }}{\frac{\partial H}{\partial
V}\left(  y,V_{i}\left(  y\right)  \right)  }\\
&  =e^{\sigma_{i}\left(  y\right)  } \left[  -\left(  e^{-\Lambda}%
\varphi\right)  _{y}-\Phi\frac{d\sigma_{i}}{dy}\left(  e^{-\Lambda}%
\varphi\right)  _{\Phi}\right. \\
&  \left.  \left.  -\frac{d\sigma_{i}}{dy}\left(  e^{-\Lambda}\varphi\right)
-\frac{dV_{i}}{dy}\left(  y\right)  \left(  e^{-\Lambda}\varphi\right)
\right]  \right|  _{V\left(  y,V_{i}\left(  y\right)  ,\Phi e^{\sigma
_{i}\left(  y\right)  }\right)  }.
\end{align*}
It then follows, using the chain rule that:%

\[
\frac{d}{dy}\left(  e^{\sigma_{i}\left(  y\right)  }e^{-\Lambda\left(
y\right)  }\varphi\left(  y,V_{i}\left(  y\right)  ,\Phi e^{\sigma_{i}\left(
y\right)  }\right)  \right)  =-F\left(  y,V_{i}\left(  y\right)  ,\Phi
e^{\sigma_{i}\left(  y\right)  }\right)  \frac{e^{\sigma_{i}\left(  y\right)
}}{\frac{\partial H}{\partial V}\left(  y,V_{i}\left(  y\right)  \right)  }.
\]
Formula (\ref{S8E8}) then becomes:
\[
J=\sum_{i=1}^{2}\left(  -1\right)  ^{i-1}\int_{y_{0}}^{\infty}\int_{0}%
^{\infty}\bar{A}_{0}\left(  \Phi\right)  \frac{d}{dy}\left(  e^{\sigma
_{i}\left(  y\right)  }e^{-\Lambda\left(  y\right)  }\varphi\left(
y,V_{i}\left(  y\right)  ,\Phi e^{\sigma_{i}\left(  y\right)  }\right)
\right)  d\Phi dy
\]
or, equivalently:
\[
J=\sum_{i=1}^{2}\left(  -1\right)  ^{i-1}\int_{0}^{\infty}\bar{A}_{0}\left(
\Phi\right)  \left(  e^{\sigma_{i}\left(  y_{0}\right)  }e^{-\Lambda\left(
y_{0}\right)  }\varphi\left(  y_{0},V_{i}\left(  y_{0}^{+}\right)  ,\Phi
e^{\sigma_{i}\left(  y_{0}\right)  }\right)  \right)  d\Phi
\]
and using (\ref{S7E6a}), (\ref{S7E6b}), (\ref{S7E8}):
\[
J\equiv\sum_{i=1}^{2}\left(  -1\right)  ^{i-1}\int_{0}^{\infty}\bar{A}%
_{0}\left(  \Phi\right)  \varphi\left(  y_{0},V_{i}\left(  y_{0}\right)
,\Phi\right)  d\Phi.
\]
Due to the fact that $\varphi\left(  y_{0},V_{i}\left(  y_{0}^{+}\right)
,\Phi\right)  =\varphi\left(  y_{0},V_{0},\Phi\right)  $ for $i=1,2$ we have
$J=0$ and (\ref{S8E5a}) follows. This concludes the proof of the theorem.
\end{proof}

\bigskip

We now remark that it is possible to derive some detailed information about
the behaviour of the curves $\gamma_{1},\;\gamma_{2}$ as $y\rightarrow\infty.$

\bigskip

\begin{theorem}
\label{Th3}Suppose that the curves $\gamma_{1},\;\gamma_{2}$ are as in Theorem
\ref{Th1}. Then, the following asymptotic formulas hold:
\begin{align*}
U  &  =\log\left(  \frac{y}{y_{0}}\right)  +\log\left(  \sqrt{1-y_{0}^{2}%
}\right)  +o\left(  1\right)  \;\text{as\ \ }y\rightarrow\infty,\\
\Lambda &  \rightarrow\log\left(  \sqrt{3}\right)  \;\text{as\ \ }%
y\rightarrow\infty,\\
V_{1}  &  =-\frac{2y_{0}\sqrt{3\left(  1-y_{0}^{2}\right)  }}{\left(
1-4y_{0}^{2}\right)  y}(1+o(1))\;\;\text{as\ \ }y\rightarrow\infty,\\
V_{2}  &  =-\frac{\sqrt{1-y_{0}^{2}}}{\sqrt{3}y_{0}}\frac{C_{1}}{y}\left(
\frac{y_{0}}{y}\right)  ^{2}(1+o(1))\;\;\text{as\ \ }y\rightarrow\infty
\end{align*}
for a suitable constant $C_{1}\in\mathbb{R.}$
\end{theorem}

Notice that the asymptotic behaviour of the solutions in Theorem \ref{Th3}
shows that the support of these solutions approaches the line $\left\{
V=0\right\}  $ away from the self-similar region (i.e. for $y\rightarrow
\infty$). This is the one of the main differences between the solutions
described in this paper and the ones in \cite{martingarcia}.

It is relevant to notice that the spacetime described by the solutions in
Theorem \ref{Th3} exhibits curvature singularities and not just coordinate
singularities. To this end we use Kretschmann scalar (cf. \cite{rendall08}):%
\[
R^{\alpha\beta\gamma\delta}R_{\alpha\beta\gamma\delta}=4K^{2}+\frac{16m^{2}%
}{r^{6}}+12r^{-2}\nabla_{a}\nabla_{b}r\nabla^{a}\nabla^{b}r
\]
where $K$ is the gaussian curvature of the quotient of the spacetime by the
symmetry group and $m$ is the Hawking mass that can be computed by means of:%
\[
m=\frac{r}{2}\left(  1-\partial_{a}r\partial^{a}r\right)  .
\]
Combining (\ref{met1}), (\ref{S3E11}), (\ref{S3E12}) we obtain the following
self-similar form for the Hawking mass:%
\[
m=\frac{r}{2}\left(  1-e^{-2\Lambda\left(  \frac{r}{\left(  -t\right)
}\right)  }\right)
\]
and therefore, it follows from Theorem \ref{Th3} that:%
\[
m\sim\frac{r}{3}\ \ \text{for\ \ }\frac{r}{\left(  -t\right)  }\text{
sufficiently large}.
\]
On the other hand, the last term in the Kretschmann scalar can be written as
(cf. \cite{DR}, Appendix A):
\[
24r^{-2}\left(  \frac{1}{2r}\left(  k-\nabla_{b}r\nabla^{c}r\right)  +2\pi
r\mathrm{tr}T\right)  ^{2}+96\pi^{2}\left(  T_{ab}-\frac{\mathrm{tr}T}%
{2}g_{ab}\right)  \left(  T^{ab}-\frac{\mathrm{tr}T}{2}g^{ab}\right)  .
\]
The last term turns out to be positive for any matter model satisfying the
dominant energy condition, which includes in particular the case of Vlasov
matter. Therefore $R^{\alpha\beta\gamma\delta}R_{\alpha\beta\gamma\delta}%
\geq\frac{16m^{2}}{r^{6}}$ and so the curvature becomes singular as
$r\rightarrow0$ for a fixed large value of $\frac{r}{\left(  -t\right)  }.$

We remark that the solutions which have been derived do not provide an example
of violation of the cosmic censorship hypothesis for Vlasov matter, because
the spacetimes concerned are not asymptotically flat as $r\rightarrow\infty.$
Moreover, it turns out that the region contained inside the light cone
reaching the singular point at $r=0,\ t=0^{-}$ in the spacetime described by
Theorem \ref{Th3} is dependent on the data on the whole region with $0\leq
r<\infty.$ This implies that a gluing of this spacetime with another one
causally disconnected \ from the singular point is not possible, because this
would require doing some gluing along regions where $r=\infty.$ In order to
check these statements it is convenient to rewrite the metric (\ref{met1}) in
double null coordinates. Notice that (\ref{met1}), (\ref{S3E11})
and (\ref{S3E12}) yield the following self-similar structure for the metric:%
\[
ds^{2}=-e^{2U\left(  \frac{r}{\left(  -t\right)  }\right)  }dt^{2}%
+e^{2\Lambda\left(  \frac{r}{\left(  -t\right)  }\right)  }dr^{2}+r^{2}\left(
d\theta^{2}+\sin^{2}\theta d\varphi^{2}\right)  .
\]
The double null coordinates are then just the constants of integration
associated to the pair of differential equations:%
\begin{align*}
-e^{U\left(  \frac{r}{\left(  -t\right)  }\right)  }dt+e^{\Lambda\left(
\frac{r}{\left(  -t\right)  }\right)  }dr  &  =0,\\
e^{U\left(  \frac{r}{\left(  -t\right)  }\right)  }dt+e^{\Lambda\left(
\frac{r}{\left(  -t\right)  }\right)  }dr  &  =0.
\end{align*}
The solutions of these equations can be written in terms of two integration
constants $u$ and $v$ that will define the double null coordinates. The
particular choice of coordinates has been made in order to obtain $u$ and $v$
taking values in compact sets:%
\begin{align*}
\mathrm{arctanh}\left(  u\right)   &  =\log\left(  -t\right)  +\int_{0}%
^{y}\frac{e^{\Lambda\left(  \xi\right)  -U\left(  \xi\right)  }}{1+\xi
e^{\Lambda\left(  \xi\right)  -U\left(  \xi\right)  }}d\xi\\
\mathrm{arctanh}\left(  v\right)   &  =\log\left(  -t\right)  -\int_{0}%
^{y}\frac{e^{\Lambda\left(  \xi\right)  -U\left(  \xi\right)  }}{1-\xi
e^{\Lambda\left(  \xi\right)  -U\left(  \xi\right)  }}d\xi
\end{align*}
In the region close to the centre (i.e. $y<<1$) the structure of the metric is
similar to Minkowski. On the other hand, Theorem \ref{Th3} yields the
following asymptotics for $r>>\left(  -t\right)  :$%
\begin{align*}
\mathrm{arctanh}\left(  v\right)   &  \sim\log\left(  -t\right)  +\frac
{\sqrt{3}y_{0}}{\sqrt{1-y_{0}^{2}}}\log\left(  \frac{r}{\left(  -t\right)
}\right)  ,\\
\mathrm{arctanh}\left(  u\right)   &  \sim-\log\left(  -t\right)  +\frac
{\sqrt{3}y_{0}}{\sqrt{1-y_{0}^{2}}}\log\left(  \frac{r}{\left(  -t\right)
}\right)  .
\end{align*}
The light cone approaching the singular point is described in these
coordinates by the line $u=1.$ Notice along such a line, for $v$ of order one
we would have $r=\infty,$ whence the assertion above follows.

For these reasons a spacetime behaving asymptotically as Minkowski cannot be
obtained gluing the self-similar solution obtained in this paper with a
spacetime causally disconnected from the singular point. This kind of gluing
might be possible for non self-similar solutions of the Einstein equations
behaving asymptotically near the singular point like those described in this
paper. However, such an analysis is beyond the scope of this paper.

\section{PROOF OF THEOREM \ref{Th1}.\label{ProT1}}

\bigskip

The strategy used to prove Theorem \ref{Th1} is the following. We first
transform the original problem (\ref{S4E3}), (\ref{S4E4}), (\ref{S7E1}),
(\ref{S7E2}), (\ref{S7E3}), (\ref{S7E4}), (\ref{S7E6})-(\ref{S7E8}),
(\ref{S7E10a}), (\ref{S7E10b}) into a family of four-dimensional autonomous
systems depending on the parameter $\beta_{0}$ by means of a change of
variables. It will be shown that proving Theorem \ref{Th1} is equivalent to
finding an orbit for this system connecting two specific points $P_{0}%
,\;P_{1}$ of the four-dimensional phase space. The point $P_{1}$ is a unstable
saddle point with an associated three-dimensional stable manifold
$\mathcal{M}=\mathcal{M}\left(  \beta_{0}\right)  $ that can be described in
detail in the limit $y_{0}\rightarrow0$. A shooting argument will show that
for a suitable choice of the parameter $\beta_{0}$ the manifold $\mathcal{M}%
\left(  \beta_{0}\right)  $ contains the point $P_{0}.$ In the rest of this
section we give the details of this argument.

\subsection{Reduction of the problem to an autonomous system.\label{reformul}}

Instead of the set of variables $\left(  y,U,\Lambda,V_{i},\sigma_{i}\right)
$ it is more convenient to use the set of variables $\left(  s,u,\Lambda
,\zeta_{i},Q_{i}\right)  $ where:%

\begin{equation}
s=\log\left(  \frac{y}{y_{0}}\right)  \;\;,\;\;U=\log\left(  \frac{y}{y_{0}%
}\right)  +u\;\;,\;\;\zeta_{i}=yV_{i}\;\;,\;\;Q_{i}=\frac{y_{0}}{y}%
e^{\sigma_{i}}\;,\;i=1,2. \label{D1E1}%
\end{equation}
Then, the evolution equations (\ref{S4E3}), (\ref{S4E4}), (\ref{S7E2a}),
(\ref{S7E7}) become:
\begin{align}
e^{u}\sqrt{\zeta_{i}^{2}+1}+y_{0}\zeta_{i}e^{\Lambda}  &  =\sqrt{1-y_{0}^{2}%
}\;\;,\;\;i=1,2\;\;,\;\;\zeta_{1}<\zeta_{2},\label{D1E6}\\
\frac{dQ_{i}}{ds}  &  =-\frac{e^{u}Q_{i}\zeta_{i}}{\left[  y_{0}e^{\Lambda
}\sqrt{\left(  \zeta_{i}\right)  ^{2}+1}+\zeta_{i}e^{u}\right]  }%
\;\;,\;\;i=1,2,\label{D1E7}\\
e^{-2\Lambda}\left(  2\Lambda_{s}-1\right)  +1  &  =\frac{\theta}{2}\left[
\frac{Q_{1}^{2}\left[  \zeta_{1}^{2}+1\right]  }{\left|  \zeta_{1}e^{u}%
+y_{0}e^{\Lambda}\sqrt{\zeta_{1}^{2}+1}\right|  }+\frac{Q_{2}^{2}\left[
\zeta_{2}^{2}+1\right]  }{\left|  \zeta_{2}e^{u}+y_{0}e^{\Lambda}\sqrt
{\zeta_{2}^{2}+1}\right|  }\right]  ,\label{D1E8}\\
e^{-2\Lambda}\left(  2u_{s}+3\right)  -1  &  =\frac{\theta}{2}\left[
\frac{Q_{1}^{2}\left(  \zeta_{1}\right)  ^{2}}{\left|  \zeta_{1}e^{u}%
+y_{0}e^{\Lambda}\sqrt{\zeta_{1}^{2}+1}\right|  }+\frac{Q_{2}^{2}\left(
\zeta_{2}\right)  ^{2}}{\left|  \zeta_{2}e^{u}+y_{0}e^{\Lambda}\sqrt{\zeta
_{2}^{2}+1}\right|  }\right]  \label{D1E9}%
\end{align}
where%
\begin{equation}
\theta=\frac{16\pi^{2}\beta_{0}}{y_{0}}. \label{D1E9a}%
\end{equation}
The initial conditions (\ref{S4E6}), (\ref{S7E8}) imply:
\begin{equation}
u=0\;\;,\;\;\Lambda=0\;\;,\;\;Q_{i}=1\;\;,\;\;i=1,2\;\;\;\text{at \ }s=0.
\label{D1E10}%
\end{equation}
Notice that the system (\ref{D1E7})-(\ref{D1E9}) with $\zeta_{i}$ as in
(\ref{D1E6}) is a four-dimensional autonomous system of equations for the
unknown functions $\left(  Q_{1},Q_{2},\Lambda,u\right)  .$ Notice however
that the system seems to becomes singular if the variables $\left(
Q_{1},Q_{2},\Lambda,u\right)  $ approach the values in (\ref{D1E10}) due to
the vanishing of the denominators in (\ref{D1E8}), (\ref{D1E9}). To treat
these singularities we rewrite the terms $\left[  y_{0}e^{\Lambda}%
\sqrt{\left(  \zeta_{i}\right)  ^{2}+1}+\zeta_{i}e^{u}\right]  .$ Notice that
(\ref{D1E6}) implies:
\begin{align}
\zeta_{i}  &  =\frac{1}{\left(  1-y_{0}^{2}e^{2\left(  \Lambda-u\right)
}\right)  }\left[  -y_{0}\sqrt{1-y_{0}^{2}}e^{\Lambda-2u}\mp Z\right]
\;,\;i=1,2,\label{L1E1}\\
Z  &  =\sqrt{\left(  e^{-2u}\left(  1-y_{0}^{2}\right)  -1\right)  \left(
1-y_{0}^{2}e^{2\left(  \Lambda-u\right)  }\right)  +y_{0}^{2}\left(
1-y_{0}^{2}\right)  e^{2\left(  \Lambda-2u\right)  }}. \label{L1E2}%
\end{align}
Then:
\[
y_{0}e^{\Lambda}\sqrt{\left(  \zeta_{i}\right)  ^{2}+1}+\zeta_{i}e^{u}=\mp
e^{u}Z\;\;,\;i=1,2
\]
and the system of equations (\ref{D1E7})-(\ref{D1E9}) becomes:
\begin{align}
\frac{dQ_{1}}{ds}  &  =\frac{Q_{1}\zeta_{1}}{Z},\;\label{L1E3}\\
\frac{dQ_{2}}{ds}  &  =-\frac{Q_{2}\zeta_{2}}{Z},\;\label{L1E4}\\
e^{-2\Lambda}\left(  2\Lambda_{s}-1\right)  +1  &  =\frac{\theta e^{-u}}%
{2}\left[  \frac{Q_{1}^{2}}{Z}\left[  \zeta_{1}^{2}+1\right]  +\frac{Q_{2}%
^{2}}{Z}\left[  \zeta_{2}^{2}+1\right]  \right]  ,\label{L1E5}\\
e^{-2\Lambda}\left(  2u_{s}+3\right)  -1  &  =\frac{\theta e^{-u}}{2}\left[
\frac{Q_{1}^{2}\zeta_{1}^{2}}{Z}+\frac{Q_{2}^{2}\zeta_{2}^{2}}{Z}\right]  .
\label{L1E6}%
\end{align}
We now eliminate the variables $\Lambda,\;u$ in (\ref{D1E7})-(\ref{D1E9}) and
replace them by the functions $Z$ and $G$ where $Z$ is as in (\ref{L1E2}) and
$G$ is defined by means of:
\begin{equation}
G=e^{-2\Lambda}. \label{L1E6a}%
\end{equation}
Then (\ref{L1E5}) becomes:
\begin{equation}
G_{s}=1-G-\frac{\theta e^{-u}}{2}\left[  \frac{Q_{1}^{2}}{Z}\left[  \zeta
_{1}^{2}+1\right]  +\frac{Q_{2}^{2}}{Z}\left[  \zeta_{2}^{2}+1\right]
\right]  . \label{L1E7}%
\end{equation}
On the other hand (\ref{L1E2}) implies:
\begin{equation}
e^{-2u}=\frac{Z^{2}+1}{\left[  \left(  1-y_{0}^{2}\right)  +y_{0}%
^{2}e^{2\Lambda}\right]  }=\frac{\left(  Z^{2}+1\right)  G}{\left[
G+y_{0}^{2}\left(  1-G\right)  \right]  } \label{L1E7a}%
\end{equation}
whence:
\[
u=-\frac{1}{2}\log\left(  \frac{\left(  Z^{2}+1\right)  G}{\left[  G+y_{0}%
^{2}\left(  1-G\right)  \right]  }\right)  .
\]
Differentiating this formula we obtain:
\[
u_{s}=-\frac{ZZ_{s}}{\left(  Z^{2}+1\right)  }-\frac{y_{0}^{2}}{2}\frac{G_{s}%
}{\left[  G+y_{0}^{2}\left(  1-G\right)  \right]  G}.
\]
Eliminating $u_{s}$ from this formula using (\ref{L1E6}), (\ref{L1E7}) we
obtain:
\begin{equation}
ZZ_{s}=\left(  \frac{3}{2}-\frac{1}{2G}-\Delta\right)  \left(  Z^{2}+1\right)
\label{L1E8}%
\end{equation}
where:%
\[
\Delta\equiv\frac{y_{0}^{2}}{2}\frac{G_{s}}{\left[  G+y_{0}^{2}\left(
1-G\right)  \right]  G}+\frac{\theta e^{-u}}{4G}\left[  \frac{Q_{1}^{2}\left(
\zeta_{1}\right)  ^{2}}{Z}+\frac{Q_{2}^{2}\left(  \zeta_{2}\right)  ^{2}}%
{Z}\right]  .
\]
Using (\ref{L1E7}) it then follows, after some computations, that:%
\begin{align}
&  4GZ\left[  G+y_{0}^{2}\left(  1-G\right)  \right]  \Delta=2\left(
1-G\right)  y_{0}^{2}Z\label{L1E8a}\\
&  +\theta e^{-u}\left[  -y_{0}^{2}\left[  Q_{1}^{2}\left[  \zeta_{1}%
^{2}+1\right]  +Q_{2}^{2}\left[  \zeta_{2}^{2}+1\right]  \right]  \right.
\nonumber\\
&  \left.  +\left[  Q_{1}^{2}\left(  \zeta_{1}\right)  ^{2}+Q_{2}^{2}\left(
\zeta_{2}\right)  ^{2}\right]  \left[  G+y_{0}^{2}\left(  1-G\right)  \right]
\right]  .\nonumber
\end{align}
The last bracket in (\ref{L1E8a}) can be rewritten as:%
\begin{align}
&  \left[  -y_{0}^{2}\left[  Q_{1}^{2}\left[  \zeta_{1}^{2}+1\right]
+Q_{2}^{2}\left[  \zeta_{2}^{2}+1\right]  \right]  +\left[  Q_{1}^{2}\left(
\zeta_{1}\right)  ^{2}+Q_{2}^{2}\left(  \zeta_{2}\right)  ^{2}\right]  \left[
G+y_{0}^{2}\left(  1-G\right)  \right]  \right] \nonumber\label{L1E8b}\\
&  =Q_{1}^{2}\left[  \zeta_{1}^{2}-y_{0}^{2}\left(  \zeta_{1}^{2}+1\right)
\right]  +Q_{2}^{2}\left[  \zeta_{2}^{2}-y_{0}^{2}\left(  \zeta_{2}%
^{2}+1\right)  \right] \nonumber\\
&  +\left(  1-y_{0}^{2}\right)  \left[  Q_{1}^{2}\left(  \zeta_{1}\right)
^{2}+Q_{2}^{2}\left(  \zeta_{2}\right)  ^{2}\right]  \left(  G-1\right)  .
\end{align}
Using (\ref{L1E1}) we obtain:%
\begin{align}
&  \left[  \zeta_{i}^{2}-y_{0}^{2}\left(  \zeta_{i}^{2}+1\right)  \right]
=\frac{\left(  1-y_{0}^{2}\right)  Z^{2}}{\left(  1-y_{0}^{2}e^{2\left(
\Lambda-u\right)  }\right)  ^{2}}\pm\frac{2y_{0}\left(  1-y_{0}^{2}\right)
^{\frac{3}{2}}e^{\Lambda-2u}}{\left(  1-y_{0}^{2}e^{2\left(  \Lambda-u\right)
}\right)  ^{2}}Z\nonumber\\
&  +y_{0}^{2}\left[  \frac{\left(  1-y_{0}^{2}\right)  ^{2}e^{2\left(
\Lambda-2u\right)  }}{\left(  1-y_{0}^{2}e^{2\left(  \Lambda-u\right)
}\right)  ^{2}}-1\right]  \;\;,\;\;i=1,2. \label{L1E8c}%
\end{align}
Plugging (\ref{L1E8c}) into (\ref{L1E8b}) it then follows that:%
\begin{align*}
&  \left[  -y_{0}^{2}\left[  Q_{1}^{2}\left[  \zeta_{1}^{2}+1\right]
+Q_{2}^{2}\left[  \zeta_{2}^{2}+1\right]  \right]  +\left[  Q_{1}^{2}\left(
\zeta_{1}\right)  ^{2}+Q_{2}^{2}\left(  \zeta_{2}\right)  ^{2}\right]  \left[
G+y_{0}^{2}\left(  1-G\right)  \right]  \right] \\
&  =\frac{\left(  1-y_{0}^{2}\right)  Z^{2}}{\left(  1-y_{0}^{2}e^{2\left(
\Lambda-u\right)  }\right)  ^{2}}\left(  Q_{1}^{2}+Q_{2}^{2}\right)
+\frac{2y_{0}\left(  1-y_{0}^{2}\right)  ^{\frac{3}{2}}e^{\Lambda-2u}}{\left(
1-y_{0}^{2}e^{2\left(  \Lambda-u\right)  }\right)  ^{2}}\left(  Q_{1}%
^{2}-Q_{2}^{2}\right)  Z\\
&  +y_{0}^{2}\left[  \frac{\left(  1-y_{0}^{2}\right)  ^{2}e^{2\left(
\Lambda-2u\right)  }}{\left(  1-y_{0}^{2}e^{2\left(  \Lambda-u\right)
}\right)  ^{2}}-1\right]  \left(  Q_{1}^{2}+Q_{2}^{2}\right) \\
&  +\left(  1-y_{0}^{2}\right)  \left[  Q_{1}^{2}\left(  \zeta_{1}\right)
^{2}+Q_{2}^{2}\left(  \zeta_{2}\right)  ^{2}\right]  \left(  G-1\right)
\end{align*}
and using (\ref{L1E8a}) we arrive at:
\begin{align}
&  \Delta=\frac{\left(  1-G\right)  y_{0}^{2}}{2G\left[  G+y_{0}^{2}\left(
1-G\right)  \right]  }\nonumber\\
&  +\frac{\theta e^{-u}}{4G\left[  G+y_{0}^{2}\left(  1-G\right)  \right]
}\left[  \frac{\left(  1-y_{0}^{2}\right)  Z}{\left(  1-y_{0}^{2}e^{2\left(
\Lambda-u\right)  }\right)  ^{2}}\left(  Q_{1}^{2}+Q_{2}^{2}\right)  \right.
\nonumber\\
&  \left.  +\frac{2y_{0}\left(  1-y_{0}^{2}\right)  ^{\frac{3}{2}}
e^{\Lambda-2u}}{\left(  1-y_{0}^{2}e^{2\left(  \Lambda-u\right)  }\right)
^{2}}\left(  Q_{1}^{2}-Q_{2}^{2}\right)  +\frac{1}{Z}\Phi\right]
\label{L1E8d}%
\end{align}
where:%
\begin{align}
&  \Phi=y_{0}^{2}\left[  \frac{\left(  1-y_{0}^{2}\right)  ^{2}e^{2\left(
\Lambda-2u\right)  }}{\left(  1-y_{0}^{2}e^{2\left(  \Lambda-u\right)
}\right)  ^{2}}-1\right]  \left(  Q_{1}^{2}+Q_{2}^{2}\right) \nonumber\\
&  +\left(  1-y_{0}^{2}\right)  \left[  Q_{1}^{2}\left(  \zeta_{1}\right)
^{2}+Q_{2}^{2}\left(  \zeta_{2}\right)  ^{2}\right]  \left(  G-1\right)  .
\label{L1E9}%
\end{align}
In order to obtain analytic solutions it is convenient to introduce the change
of variables:%
\begin{equation}
ds=2GZd\chi\;\;,\;\;\chi=0\;\;\text{at\ \ }s=0. \label{L1E9aa}%
\end{equation}
Then the system (\ref{L1E3}), (\ref{L1E4}), (\ref{L1E7}), (\ref{L1E8})
becomes:%
\begin{align}
\frac{dQ_{1}}{d\chi}  &  =2GQ_{1}\zeta_{1},\label{F1E1}\\
\frac{dQ_{2}}{d\chi}  &  =-2GQ_{2}\zeta_{2},\label{F1E2}\\
\frac{dG}{d\chi}  &  =2G\left[  Z\left(  1-G\right)  -\frac{\theta e^{-u}}%
{2}\left[  Q_{1}^{2}\left[  \zeta_{1}^{2}+1\right]  +Q_{2}^{2}\left[
\zeta_{2}^{2}+1\right]  \right]  \right]  ,\label{F1E3}\\
\frac{dZ}{d\chi}  &  =\left(  3G-1-2G\Delta\right)  \left(  Z^{2}+1\right)
\label{F1E4}%
\end{align}
with the initial conditions:%
\begin{equation}
Q_{1}=Q_{2}=1\;\;,\;\;G=1\;\;,\;\;Z=0\;\;,\;\text{at}\;\chi=0. \label{L1E9e}%
\end{equation}
We can further simplify $\Phi$ in (\ref{L1E9}) using (\ref{L1E1}):
\begin{align}
\Phi &  =y_{0}^{2}\left[  \frac{\left(  1-y_{0}^{2}\right)  ^{2}e^{2\left(
\Lambda-2u\right)  }}{\left(  1-y_{0}^{2}e^{2\left(  \Lambda-u\right)
}\right)  ^{2}}-1\right]  \left(  Q_{1}^{2}+Q_{2}^{2}\right) \label{L1E9f}\\
&  +\frac{\left(  1-y_{0}^{2}\right)  \left(  G-1\right)  }{\left(
1-y_{0}^{2}e^{2\left(  \Lambda-u\right)  }\right)  ^{2}}\left[  \left[
y_{0}^{2}\left(  1-y_{0}^{2}\right)  e^{2\left(  \Lambda-2u\right)  }%
+Z^{2}\right]  \left(  Q_{1}^{2}+Q_{2}^{2}\right)  \right] \nonumber\\
&  \left.  +2y_{0}\sqrt{1-y_{0}^{2}}Ze^{\Lambda-2u}\left(  Q_{1}^{2}-Q_{2}%
^{2}\right)  \right]  .\nonumber
\end{align}
In order to identify the behaviour of $\Phi$ as $Z\rightarrow0$ we write the
terms in brackets on the right-hand side of (\ref{L1E9f}) as:%
\begin{align*}
&  \left[  \frac{\left(  1-y_{0}^{2}\right)  ^{2}e^{2\left(  \Lambda
-2u\right)  }}{\left(  1-y_{0}^{2}e^{2\left(  \Lambda-u\right)  }\right)
^{2}}-1\right] \\
&  =\frac{1}{\left(  1-y_{0}^{2}e^{2\left(  \Lambda-u\right)  }\right)  ^{2}%
}\left[  \left(  1-y_{0}^{2}\right)  ^{2}\left(  e^{2\left(  \Lambda
-2u\right)  }-1\right)  +2\left(  1-y_{0}^{2}\right)  y_{0}^{2}\left(
e^{2\left(  \Lambda-u\right)  }-1\right)  \right. \\
&  \left.  -y_{0}^{4}\left(  e^{2\left(  \Lambda-u\right)  }-1\right)
^{2}\right]  .
\end{align*}
Then (\ref{L1E9f}) becomes:%
\begin{align}
&  \Phi=\frac{y_{0}^{2}\left(  1-y_{0}^{2}\right)  \left(  Q_{1}^{2}+Q_{2}%
^{2}\right)  }{\left(  1-y_{0}^{2}e^{2\left(  \Lambda-u\right)  }\right)
^{2}}\left[  \left(  1-y_{0}^{2}\right)  \left(  e^{2\left(  \Lambda
-2u\right)  }-1\right)  +2y_{0}^{2}\left(  e^{2\left(  \Lambda-u\right)
}-1\right)  \right. \nonumber\\
&  \left.  +\left(  1-y_{0}^{2}\right)  \left(  G-1\right)  e^{2\left(
\Lambda-2u\right)  }\right] \nonumber\\
&  +\frac{\left(  1-y_{0}^{2}\right)  \left(  G-1\right)  }{\left(
1-y_{0}^{2}e^{2\left(  \Lambda-u\right)  }\right)  ^{2}}\left[  Z^{2}\left(
Q_{1}^{2}+Q_{2}^{2}\right)  +2y_{0}\sqrt{1-y_{0}^{2}}Ze^{\Lambda-2u}\left(
Q_{1}^{2}-Q_{2}^{2}\right)  \right] \nonumber\\
&  -\frac{y_{0}^{6}\left(  Q_{1}^{2}+Q_{2}^{2}\right)  }{\left(  1-y_{0}%
^{2}e^{2\left(  \Lambda-u\right)  }\right)  ^{2}}\left(  e^{2\left(
\Lambda-u\right)  }-1\right)  ^{2}. \label{L1E9g}%
\end{align}
In order to simplify this formula we write, using (\ref{L1E6a}),
(\ref{L1E7a}):
\begin{align*}
&  \left[  \left(  1-y_{0}^{2}\right)  \left(  e^{2\left(  \Lambda-2u\right)
}-1\right)  +2y_{0}^{2}\left(  e^{2\left(  \Lambda-u\right)  }-1\right)
+\left(  1-y_{0}^{2}\right)  \left(  G-1\right)  e^{2\left(  \Lambda
-2u\right)  }\right] \\
&  =-\left(  1-y_{0}^{2}\right)  \left(  1-e^{-4u}\right)  +2y_{0}^{2}\left(
e^{2\left(  \Lambda-u\right)  }-1\right) \\
&  =-\left(  1-y_{0}^{2}\right)  \left(  1-\left(  \frac{G}{\left[
G+y_{0}^{2}\left(  1-G\right)  \right]  }\right)  ^{2}\right)  +2y_{0}%
^{2}\left(  \frac{1}{\left[  G+y_{0}^{2}\left(  1-G\right)  \right]
}-1\right) \\
&  +\left(  1-y_{0}^{2}\right)  \left(  2Z^{2}+Z^{4}\right)  \left(  \frac
{G}{\left[  G+y_{0}^{2}\left(  1-G\right)  \right]  }\right)  ^{2}+2y_{0}%
^{2}\left(  \frac{Z^{2}}{\left[  G+y_{0}^{2}\left(  1-G\right)  \right]
}\right),
\end{align*}%
\[
\left(  e^{2\left(  \Lambda-u\right)  }-1\right)  =\frac{Z^{2}+\left(
1-G\right)  \left(  1-y_{0}^{2}\right)  }{\left[  G+y_{0}^{2}\left(
1-G\right)  \right]  }.
\]
Plugging these formulas into (\ref{L1E9g}) we obtain, after some computations:%
\begin{align}
&  \frac{\Phi}{Z} =\frac{y_{0}^{2}\left(  1-y_{0}^{2}\right)  \left(
Q_{1}^{2}+Q_{2}^{2}\right)  }{\left(  1-y_{0}^{2}e^{2\left(  \Lambda-u\right)
}\right)  ^{2}}\left[  \left(  1-y_{0}^{2}\right)  \left(  2Z+Z^{3}\right)
\left(  \frac{G}{\left[  G+y_{0}^{2}\left(  1-G\right)  \right]  }\right)
^{2}\right. \nonumber\\
&  \left.  +2y_{0}^{2}\left(  \frac{Z}{\left[  G+y_{0}^{2}\left(  1-G\right)
\right]  }\right)  \right] \nonumber\\
&  +\frac{\left(  1-y_{0}^{2}\right)  \left(  G-1\right)  }{\left(
1-y_{0}^{2}e^{2\left(  \Lambda-u\right)  }\right)  ^{2}}\left[  Z\left(
Q_{1}^{2}+Q_{2}^{2}\right)  +2y_{0}\sqrt{1-y_{0}^{2}}e^{\Lambda-2u}\left(
Q_{1}^{2}-Q_{2}^{2}\right)  \right] \nonumber\\
&  -\frac{y_{0}^{6}\left(  Q_{1}^{2}+Q_{2}^{2}\right)  }{\left(  1-y_{0}%
^{2}e^{2\left(  \Lambda-u\right)  }\right)  ^{2}}\left[  \frac{2\left(
1-y_{0}^{2}\right)  Z\left(  1-G\right)  +Z^{3}}{\left[  G+y_{0}^{2}\left(
1-G\right)  \right]  ^{2}}\right]  . \label{L1E9h}%
\end{align}

\bigskip

Summarizing, we have transformed the original problem (\ref{S4E3}),
(\ref{S4E4}), (\ref{S7E2a}), (\ref{S7E7}) into the system of equations
(\ref{F1E1})-(\ref{F1E4}) with $\Delta$ as in (\ref{L1E8d}), $\frac{\Phi}{Z}$
as in (\ref{L1E9h}), $\zeta_{i}$ as in (\ref{L1E1}) and $\Lambda,u$ given by
(\ref{L1E6a}), (\ref{L1E7a}). The initial data for $\left(  Q_{1}%
,Q_{2},G,Z\right)  $ are as in (\ref{L1E9e}).

Some of the forms that we have derived for the ODE problems above are more
convenient for describing the solutions in different regions of the phase
space. We will change freely between the different groups of equivalent
variables in the following.

\subsection{Local existence of the curves $\gamma_{1},\;\gamma_{2}.$}

\bigskip

With the reformulation of the problem obtained in the previous subsection the
existence of the curves $\gamma_{1},\;\gamma_{2}$ in a neighbourhood of the
point $\left(  y_{0},V_{0}\right)  $ can be obtained using standard ODE theory.

\bigskip

\begin{proposition}
\label{loctraj}For any $y_{0}\in\left(  0,1\right)  $ and any $\beta_{0}>0$
there exist $\delta>0$ and two curves $\gamma_{1},\;\gamma_{2}$ that can be
parametrized as%
\begin{equation}
\gamma_{i}=\left\{  \left(  y,V\right)  :y_{0}<y<y_{0}+\delta\;,\;V=V_{i}%
\left(  y\right)  \right\}  \;\;,\;\;i=1,2 \label{J1E1}%
\end{equation}
with the functions $V_{1}\left(  y\right)  ,\;V_{2}\left(  y\right)  $ as in
(\ref{S7E3}), (\ref{S7E4}) satisfying (\ref{S7E1}), (\ref{S7E2}), (\ref{S7E6})
the functions $U,\;\Lambda$ satisfying (\ref{S4E3}), (\ref{S4E4}) and
(\ref{S7E6a}), (\ref{S7E6b}) with $\tilde{\rho},\;\tilde{p}$ as in
(\ref{S7E10a}), (\ref{S7E10b}) and $\sigma_{1},\;\sigma_{2}$ solving
(\ref{S7E7}), (\ref{S7E8}).
\end{proposition}

\noindent

\begin{proof}
The arguments in Subsection \ref{reformul} show that the proposition follows
from proving local existence and uniqueness for (\ref{F1E1})-(\ref{F1E4}) with
initial data (\ref{L1E9e}). Since the right-hand side of (\ref{F1E1}%
)-(\ref{F1E4}) is analytic in a neighbourhood of $\left(  Q_{1},Q_{2}%
,G,Z\right)  =\left(  1,1,1,0\right)  $ it follows that there exists a unique
solution of (\ref{L1E9e}), (\ref{F1E1})-(\ref{F1E4}) on an interval of the
form $0<\chi<\delta_{0}$ for some $\delta_{0}>0.$ Moreover, for such a
solution $\Delta\rightarrow0$ as $\chi\rightarrow0^{+},$ whence $Z\sim2\chi$
as $\chi\rightarrow0^{+}.$ Therefore (\ref{L1E9aa}) yields:
\[
s\sim2\chi^{2}\;\;\text{as\ \ }\chi\rightarrow0^{+}\;\;,\;\;\chi\sim
\sqrt{\frac{s}{2}}\;\;\text{as\ \ }s\rightarrow0^{+},
\]%
\begin{equation}
Z\sim\sqrt{2s}\;\;\mathrm{as}\ \ s\rightarrow0^{+}. \label{J2E1}%
\end{equation}
Using (\ref{D1E1}) it follows that:%
\begin{equation}
s\sim\frac{y-y_{0}}{y_{0}}\;\;\text{as\ \ }y\rightarrow y_{0}^{+}.
\label{J2E2}%
\end{equation}
Combining then (\ref{D1E1}) and (\ref{L1E1}) we obtain (\ref{S7E6}). The
asymptotics (\ref{S7E6a}), (\ref{S7E6b}) follows from the asymptotics for
$G,\;Z$ in an analogous way.
\end{proof}

\bigskip

\noindent Moreover, we can prove Proposition \ref{Prop1} in a similar way.

\noindent

\begin{proof}
[Proof of Proposition \ref{Prop1}]It follows from (\ref{S5E3}), (\ref{D1E1}),
(\ref{J2E1}), (\ref{J2E2}).
\end{proof}

\bigskip

\noindent We notice for further reference that we have also proved the
following result:

\bigskip

\begin{proposition}
\label{local}There exists a unique solution of the system (\ref{D1E7}%
)-(\ref{D1E9}) with $\zeta_{i}$ as in (\ref{D1E6}) and initial data $\left(
Q_{1},Q_{2},\Lambda,u\right)  =\left(  1,1,0,0\right)  $ as $s\rightarrow
0^{+}.$
\end{proposition}

\subsection{Steady states for the system (\ref{F1E1})-(\ref{F1E4}).}

In order to study the steady states of (\ref{F1E1})-(\ref{F1E4}) it is more
convenient to use the form of the equations in (\ref{D1E6})-(\ref{D1E9}). Then
the steady states are characterized by:%
\begin{align}
Q_{i}\zeta_{i}  &  =0\;\;i=1,2,\label{Z1E1}\\
-e^{-2\Lambda}+1  &  =\frac{\theta}{2}\left[  \frac{Q_{1}^{2}}{\left|
\zeta_{1}e^{u}+y_{0}e^{\Lambda}\sqrt{\zeta_{1}^{2}+1}\right|  }\left[
\zeta_{1}^{2}+1\right]  \right. \nonumber\\
&  \left.  +\frac{Q_{2}^{2}}{\left|  \zeta_{2}e^{u}+y_{0}e^{\Lambda}%
\sqrt{\zeta_{2}^{2}+1}\right|  }\left[  \left(  \zeta_{2}\right)
^{2}+1\right]  \right]  ,\label{Z1E2}\\
3e^{-2\Lambda}-1  &  =\frac{\theta}{2}\left[  \frac{Q_{1}^{2}\left(  \zeta
_{1}\right)  ^{2}}{\left|  \zeta_{1}e^{u}+y_{0}e^{\Lambda}\sqrt{\zeta_{1}%
^{2}+1}\right|  }+\frac{Q_{2}^{2}\left(  \zeta_{2}\right)  ^{2}}{\left|
\zeta_{2}e^{u}+y_{0}e^{\Lambda}\sqrt{\zeta_{2}^{2}+1}\right|  }\right]  .
\label{Z1E3}%
\end{align}
The first and third equations imply:%
\begin{equation}
3e^{-2\Lambda}-1=0. \label{Z1E4}%
\end{equation}
Then, the second equation reduces to:%
\begin{equation}
\frac{2}{3}=\frac{\theta}{2}\left[  \frac{Q_{1}^{2}}{\left|  \zeta_{1}%
e^{u}+y_{0}e^{\Lambda}\sqrt{\zeta_{1}^{2}+1}\right|  }+\frac{Q_{2}^{2}%
}{\left|  \zeta_{2}e^{u}+y_{0}e^{\Lambda}\sqrt{\zeta_{2}^{2}+1}\right|
}\right]  . \label{Z1E5}%
\end{equation}
Notice that (\ref{Z1E5}) implies that at least one of the variables
$Q_{1},\;Q_{2}$ is different from zero at the steady state. Suppose that both
of them are different from zero. Then $\zeta_{1}=\zeta_{2}=0,$ whence, using%
\[
e^{u}\sqrt{\zeta_{i}^{2}+1}+y_{0}\zeta_{i}e^{\Lambda}=\sqrt{1-y_{0}^{2}%
}\;\;,\;\;i=1,2\;\;,\;\;\zeta_{1}\leq\zeta_{2}%
\]
it follows that:%
\begin{equation}
e^{u}=\sqrt{1-y_{0}^{2}} \label{Z1E6}%
\end{equation}
and (\ref{Z1E5}) reduces to:%
\[
\left(  Q_{1}^{2}+Q_{2}^{2}\right)  =\frac{4y_{0}e^{\Lambda}}{3\theta}%
=\frac{4y_{0}\sqrt{3}}{3\theta}.
\]

This defines a family of steady states. Local analysis near these solutions
indicates that they are reached for finite values of $y.$ Since we are
interested in solutions defined for arbitrarily large values of $y>y_{0}$ a
more detailed analysis of these solutions will not be pursued here. We will
then restrict our analysis to the solutions for which $Q_{1}Q_{2}=0$.

Suppose that $Q_{1}\neq0$. Then $\zeta_{1}=0$. (\ref{Z1E6}) implies:
\begin{align*}
\sqrt{\zeta_{2}^{2}+1}+\frac{y_{0}\zeta_{2}e^{\Lambda}}{\sqrt{1-y_{0}^{2}}}
&  =1,\\
\zeta_{2}  &  =\frac{\sqrt{1-y_{0}^{2}}}{y_{0}e^{\Lambda}}\left[
1-\sqrt{\zeta_{2}^{2}+1}\right]  <0.
\end{align*}
This contradicts $\zeta_{1}\leq\zeta_{2}.$ Therefore for solutions with
$Q_{1}Q_{2}=0$ we must have $Q_{2}\neq0$ whence $\zeta_{2}=0.$ Then
(\ref{Z1E6}) is satisfied and (\ref{Z1E5}) yields:%
\[
Q_{2}=\sqrt{\frac{4\sqrt{3}}{3\theta}y_{0}}=\frac{2\sqrt{y_{0}}}{3^{\frac
{1}{4}}\sqrt{\theta}}.
\]
We remark that for this solution:%
\[
\zeta_{1}=-\frac{2he^{\Lambda_{\infty}}}{\left(  h^{2}-e^{2\Lambda_{\infty}%
}\right)  }=-\frac{2y_{0}\sqrt{3\left(  1-y_{0}^{2}\right)  }}{1-4y_{0}^{2}}.
\]
In order to have $\zeta_{1}<\zeta_{2}=0$ we need $y_{0}\in\left(  0,\frac
{1}{2}\right)  .$

Summarizing, for each $y_{0}\in\left(  0,\frac{1}{2}\right)  $ the system
(\ref{F1E1})-(\ref{F1E4}) has the following steady state:%
\begin{align}
Q_{1}  &  =Q_{1,\infty}=0,\label{Z1E7a}\\
Q_{2}  &  =Q_{2,\infty}=\frac{2\sqrt{y_{0}}}{3^{\frac{1}{4}}\sqrt{\theta}%
,}\label{Z1E7b}\\
\Lambda &  =\Lambda_{\infty}=\frac{\log\left(  3\right)  }{2},\label{Z1E7c}\\
u  &  =u_{\infty}=\log\left(  \sqrt{1-y_{0}^{2}}\right)  . \label{Z1E7d}%
\end{align}
We also introduce the following notation for further reference:%

\begin{align}
\zeta_{1,\infty}  &  =-\frac{2he^{\Lambda_{\infty}}}{\left(  h^{2}%
-e^{2\Lambda_{\infty}}\right)  }=-\frac{2y_{0}\sqrt{3\left(  1-y_{0}%
^{2}\right)  }}{1-4y_{0}^{2}},\label{Z1E8a}\\
\zeta_{2,\infty}  &  =0. \label{Z1E8b}%
\end{align}

\subsection{Linearization near the equilibrium.}

\bigskip The main result that we prove in this subsection is the following:

\begin{theorem}
\label{linear}For each $y_{0}\in\left(  0,\frac{1}{2}\right)  $\thinspace the
point $P_{1}=\left(  Q_{1,\infty},Q_{2,\infty},\Lambda_{\infty},u_{\infty
}\right)  $ defined by (\ref{Z1E7a})-(\ref{Z1E7d}) is an unstable hyperbolic
point of the system (\ref{D1E6})-(\ref{D1E9}). The corresponding stable
manifold of the point $\left(  Q_{1,\infty},Q_{2,\infty},\Lambda_{\infty
},u_{\infty}\right)  $ that will be denoted by $\mathcal{M}_{\theta}$ is
three-dimensional and it is tangent at this point to the subspace generated by
the vectors%
\begin{equation}
\left\{  \left(
\begin{array}
[c]{c}%
1\\
0\\
0\\
0
\end{array}
\right)  ,\;\left(
\begin{array}
[c]{c}%
0\\
-\frac{\left(  1-y_{0}^{2}\right)  }{3^{\frac{5}{4}}y_{0}^{\frac{3}{2}}%
\sqrt{\theta}}\\
-\frac{2}{3}\\
1
\end{array}
\right)  ,\;\left(
\begin{array}
[c]{c}%
0\\
-\frac{2\sqrt{1-y_{0}^{2}}}{3^{\frac{5}{4}}y_{0}^{\frac{3}{2}}\sqrt{\theta}}\\
-\frac{\sqrt{1-y_{0}^{2}}}{3y_{0}}\\
1
\end{array}
\right)  \right\}  .\label{LS1}%
\end{equation}

\end{theorem}

\noindent

\begin{proof}
The key ingredient in the proof of this theorem is the linearization of the
system (\ref{D1E6})-(\ref{D1E9}) around the point $\left(  Q_{1,\infty
},Q_{2,\infty},\Lambda_{\infty},u_{\infty}\right)  .$ Let us write:%
\begin{align*}
\Lambda &  =\Lambda_{\infty}+L,\\
u  &  =u_{\infty}+\nu,\\
Q_{1}  &  =Q_{1,\infty}+q_{1}=q_{1},\\
Q_{2}  &  =Q_{2,\infty}+q_{2}.
\end{align*}
Neglecting terms quadratic in $\left\vert L\right\vert +\left\vert
\nu\right\vert +\left\vert q_{1}\right\vert +\left\vert q_{2}\right\vert $ we
obtain, after some tedious, but mechanical computations, the following
linearized problem:%
\begin{align}
\frac{dq_{1}}{ds}  &  =-\frac{2h^{2}}{\left(  h^{2}-3\right)  }q_{1}%
=-\frac{2\left(  1-y_{0}^{2}\right)  }{\left(  1-4y_{0}^{2}\right)  }%
q_{1},\label{X1E1}\\
\frac{dq_{2}}{ds}  &  =\frac{2\left(  1-y_{0}^{2}\right)  }{3^{\frac{5}{4}%
}\sqrt{\theta}y_{0}^{\frac{3}{2}}}\nu,\label{X1E2}\\
L_{s}  &  =\frac{3^{\frac{1}{4}}\sqrt{\theta}}{\sqrt{y_{0}}}q_{2}%
+\frac{\left(  1-y_{0}^{2}\right)  }{3y_{0}^{2}}\nu-2L,\label{X1E3}\\
\nu_{s}  &  =3L. \label{X1E4}%
\end{align}
Looking for solutions of the linearized problem with the form:%
\[
e^{\gamma s}\left(
\begin{array}
[c]{c}%
A_{1}\\
A_{2}\\
A_{3}\\
A_{4}%
\end{array}
\right)
\]
we obtain the following possible values of $\gamma$ with their corresponding
eigenvectors:%
\[
\gamma_{1}=-\frac{2\left(  1-y_{0}^{2}\right)  }{\left(  1-4y_{0}^{2}\right)
}\;\leftrightarrow\left(
\begin{array}
[c]{c}%
A_{1}\\
A_{2}\\
A_{3}\\
A_{4}%
\end{array}
\right)  =\left(
\begin{array}
[c]{c}%
1\\
0\\
0\\
0
\end{array}
\right)  ,
\]%
\[
\gamma_{2}=-2\;\leftrightarrow\left(
\begin{array}
[c]{c}%
A_{1}\\
A_{2}\\
A_{3}\\
A_{4}%
\end{array}
\right)  =\left(
\begin{array}
[c]{c}%
0\\
-\frac{\left(  1-y_{0}^{2}\right)  }{3^{\frac{5}{4}}y_{0}^{\frac{3}{2}}%
\sqrt{\theta}}\\
-\frac{2}{3}\\
1
\end{array}
\right)  ,
\]%
\[
\gamma_{3}=-\frac{\sqrt{\left(  1-y_{0}^{2}\right)  }}{y_{0}}\;\leftrightarrow
\left(
\begin{array}
[c]{c}%
A_{1}\\
A_{2}\\
A_{3}\\
A_{4}%
\end{array}
\right)  =\left(
\begin{array}
[c]{c}%
0\\
-\frac{2\sqrt{1-y_{0}^{2}}}{3^{\frac{5}{4}}y_{0}^{\frac{3}{2}}\sqrt{\theta}}\\
-\frac{\sqrt{1-y_{0}^{2}}}{3y_{0}}\\
1
\end{array}
\right)  ,
\]%
\[
\gamma_{4}=\frac{\sqrt{\left(  1-y_{0}^{2}\right)  }}{y_{0}}\;\leftrightarrow
\left(
\begin{array}
[c]{c}%
0\\
\frac{2\sqrt{1-y_{0}^{2}}}{3^{\frac{5}{4}}y_{0}^{\frac{3}{2}}\sqrt{\theta}}\\
\frac{\sqrt{1-y_{0}^{2}}}{3y_{0}}\\
1
\end{array}
\right)  .
\]
The theorem then follows from standard results for stable manifolds (cf. for
instance \cite{CLW}, \cite{P}).
\end{proof}

\subsection{Reformulation of the solution in the original
variables.\label{refTraj}}

\bigskip

Our goal now is to obtain a trajectory connecting the point $\left(
Q_{1},Q_{2},\Lambda,u\right)  =\left(  1,1,0,0\right)  $ at $s=0$ with the
point $P_{1}$ at $s=\infty$ for a suitable value of $\theta$ (or equivalently
$\beta_{0}$). Let us remark that such a trajectory would satisfy the
requirements in Theorem \ref{Th1}. Indeed, notice that such a trajectory
behaves near the point $\left(  y_{0},V_{0}\right)  $ as stated in Theorem
\ref{Th1} due to Proposition \ref{loctraj}. On the other hand, such a
trajectory would belong to the stable manifold of the point $P_{1}$ and
therefore its asymptotic behaviour as $s\rightarrow\infty$ would be given by:%
\[
\left(
\begin{array}
[c]{c}%
Q_{1}\\
Q_{2}\\
\Lambda\\
u
\end{array}
\right)  \sim\left(
\begin{array}
[c]{c}%
Q_{1,\infty}\\
Q_{2,\infty}\\
\Lambda_{\infty}\\
u_{\infty}%
\end{array}
\right)  +C_{1}e^{-2s}\left(
\begin{array}
[c]{c}%
0\\
\frac{\left(  1-y_{0}^{2}\right)  }{3^{\frac{5}{4}}\sqrt{\theta}y_{0}%
^{\frac{3}{2}}}\\
\frac{2}{3}\\
-1
\end{array}
\right)  +C_{2}e^{-\frac{2\left(  1-y_{0}^{2}\right)  }{\left(  1-4y_{0}%
^{2}\right)  }s}\left(
\begin{array}
[c]{c}%
1\\
0\\
0\\
0
\end{array}
\right)  +...
\]
for sufficiently small $y_{0}$ (cf. \cite{CLW}). Notice that the smallness of
$y_{0}$ guarantees that the last term yields a contribution larger for
$s\rightarrow\infty$ than the first quadratic corrections if $C_{2}\neq0.$

Using (\ref{D1E1}) we obtain the following asymptotics for the original set of
variables $U,\;\Lambda,\;\sigma_{i},\;V_{i},$ $i=1,2:$%

\begin{align*}
U  &  =\log\left(  \frac{y}{y_{0}}\right)  +u\sim\log\left(  \frac{y}{y_{0}%
}\right)  +\log\left(  \sqrt{1-y_{0}^{2}}\right)  +o\left(  1\right)
\;\text{as\ \ }y\rightarrow\infty,\\
\Lambda &  \rightarrow\log\left(  \sqrt{3}\right)  \;\text{as\ \ }%
y\rightarrow\infty,\\
e^{\sigma_{1}}  &  \sim C_{2}\left(  \frac{y}{y_{0}}\right)  ^{-\frac
{1+2y_{0}^{2}}{\left(  1-4y_{0}^{2}\right)  }}\;\text{as\ \ }y\rightarrow
\infty,\\
e^{\sigma_{2}}  &  \sim Q_{2,\infty}\left(  \frac{y}{y_{0}}\right)
\;\text{as\ \ }y\rightarrow\infty,\\
V_{1}  &  \sim\frac{\zeta_{1,\infty}}{y}=-\frac{2y_{0}\sqrt{3\left(
1-y_{0}^{2}\right)  }}{\left(  1-4y_{0}^{2}\right)  y}\;\;\text{as\ \ }%
y\rightarrow\infty,\\
V_{2}  &  \sim-\frac{\sqrt{1-y_{0}^{2}}}{\sqrt{3}y_{0}}\frac{C_{1}}{y}\left(
\frac{y_{0}}{y}\right)  ^{2}\;\;\text{as\ \ }y\rightarrow\infty.
\end{align*}
in particular these formulas prove Theorem \ref{Th3}.

\bigskip

\bigskip

\subsection{The shooting argument: Approximation of the stable manifold
$\mathcal{M}_{\theta}$ for small $y_{0}.$}

\label{ManiApp}

\bigskip

Since the stable manifold $\mathcal{M}_{\theta}$ is three-dimensional we
cannot expect the point $\left(  Q_{1},Q_{2},\Lambda,u\right)  =\left(
1,1,0,0\right)  $ to belong to $\mathcal{M}_{\theta}$ for generic values of
$\theta.$ The intuitive idea of the proof which follows is to show that the 
manifold $\mathcal{M}_{\theta}$ divides the set $\left\{  0<G<1,\;Z>0,\;Q_{i}%
>0\;,\;i=1,2\right\}  $ into two different regions. If the point $\left(
1,1,0,0\right)  $ lies on different sides of $\mathcal{M}_{\theta}$ for
different values of $\theta$ then by continuity there must exist a value
$\theta^{\ast}$ of $\theta$ such that $\left(  1,1,0,0\right)  \in
\mathcal{M}_{\theta}.$ In the rest of the paper we will obtain approximations
to the manifold $\mathcal{M}_{\theta}$ for $y_{0}$ small that will show that
the point $\left(  1,1,0,0\right)  $ lies on different sides of $\mathcal{M}%
_{\theta}$ for large positive values of $\theta$ and small positive values of
$\theta.$ More precisely, the main result of this subsection is the following:

\begin{theorem}
\label{shootingTheorem} There exists $\bar{y}_{0}$ small enough such that, for
any $y_{0}$ in the interval $[0,\bar y_{0}]$ there exists $\theta^{\ast
}=\theta^{\ast}\left(  y_{0}\right)  >0$ such that $\left(  1,1,0,0\right)
\in\mathcal{M}_{\theta^{\ast}}.$
\end{theorem}

\noindent

\begin{proof}
In order to prove Theorem \ref{shootingTheorem} it is convenient to use the
coordinates $\left(  Q_{1},Q_{2},G,Z\right)  $ (cf. (\ref{L1E2}),
(\ref{L1E6a})). These variables satisfy the system of equations (\ref{F1E1}%
)-(\ref{F1E4}). The steady state $P_{1}=P_{1}\left(  y_{0}\right)  $ is given
in these coordinates by:%
\begin{equation}
P_{1}=\left(  Q_{1,\infty},Q_{2,\infty},G_{\infty},Z_{\infty}\right)  =\left(
0,\frac{2\sqrt{y_{0}}}{3^{\frac{1}{4}}\sqrt{\theta}},\frac{1}{3},\sqrt
{\frac{3y_{0}^{2}}{\left(  1-y_{0}^{2}\right)  }}\right)  .\label{V1E1}%
\end{equation}
The point $P_{1}$ depends continuously on $y_{0}$ if $y_{0}\in\left[
0,\frac{1}{2}\right]  .$ If $y_{0}=0$ the system (\ref{F1E1})-(\ref{F1E4}) becomes:%

\begin{align}
\frac{dQ_{1}}{d\zeta} &  =-2GZQ_{1},\label{T1}\\
\frac{dQ_{2}}{d\zeta} &  =-2GZQ_{2},\label{T2}\\
\frac{dG}{d\zeta} &  =2G\left[  Z\left(  1-G\right)  -\frac{\theta\left[
Z^{2}+1\right]  ^{\frac{3}{2}}}{2}\left(  Q_{1}^{2}+Q_{2}^{2}\right)  \right]
,\label{T3}\\
\frac{dZ}{d\zeta} &  =\left(  3G-1-\frac{\theta e^{-u}}{2}Z\left(  Q_{1}%
^{2}+Q_{2}^{2}\right)  \right)  \left(  Z^{2}+1\right)  .\label{T4}%
\end{align}
Theorem \ref{linear} shows that the point $P_{1}\left(  y_{0}\right)  $ is
hyperbolic for $y_{0}\in\left(  0,\frac{1}{2}\right]  $ with a
three-dimensional stable manifold $\mathcal{M}_{\theta}=\mathcal{M}_{\theta
}\left(  y_{0}\right)  .$ On the other hand two of the eigenvalues associated
to the linearization around $P_{1}$ of the system (\ref{F1E1})-(\ref{F1E4})
degenerate for $y_{0}=0.$ More precisely, let us write $G=\frac{1}{3}+g.$
Since $P_{1}\left(  0\right)  =\left(  0,0,\frac{1}{3},0\right)  $ we obtain
the following linearization of (\ref{T1})-(\ref{T4}) near $P_{1}\left(
0\right)  $:
\[
\frac{dQ_{1}}{d\zeta}=0\ \ ,\ \ \frac{dQ_{2}}{d\zeta}=0\ \ ,\ \ \frac
{dG}{d\zeta}=\frac{4Z}{9}\ \ ,\ \ \frac{dZ}{d\zeta}=3g.
\]
The corresponding eigenvalues are $\left\{  0,0,-\frac{2\sqrt{3}}{3}%
,\frac{2\sqrt{3}}{3}\right\}  $ and the corresponding eigenvectors are
$\left\{  \left(
\begin{array}
[c]{c}%
1\\
0\\
0\\
0
\end{array}
\right)  ,\left(
\begin{array}
[c]{c}%
0\\
1\\
0\\
0
\end{array}
\right)  ,\left(
\begin{array}
[c]{c}%
0\\
0\\
-\frac{2\sqrt{3}}{9}\\
1
\end{array}
\right)  ,\left(
\begin{array}
[c]{c}%
0\\
0\\
\frac{2\sqrt{3}}{9}\\
1
\end{array}
\right)  \right\}  .$ Standard results (cf. \cite{CLW}) show the existence of
a centre-stable manifold that will be denoted by $\mathcal{M}_{\theta}\left(
0\right)  $ that is invariant under the flow defined by the system
(\ref{T1})-(\ref{T4}) and is tangent at $P_{1}\left(  0\right)  $ to the plane
spanned by $\left\{  \left(
\begin{array}
[c]{c}%
1\\
0\\
0\\
0
\end{array}
\right)  ,\left(
\begin{array}
[c]{c}%
0\\
1\\
0\\
0
\end{array}
\right)  ,\left(
\begin{array}
[c]{c}%
0\\
0\\
-\frac{2\sqrt{3}}{9}\\
1
\end{array}
\right)  \right\}  .$

Classical results (cf. \cite{CLW}) then show that it is possible to obtain a
continuously differentiable four-dimensional manifold $\mathcal{M}%
_{\theta,{\rm ext}}\subset\left[  0,\frac{1}{2}\right]  \times\mathbb{R}^{4},$ 
with
$\left(  y_{0},Q_{1},Q_{2},G,Z\right)  \in\mathcal{M}_{{\rm ext}}$ such that:%
\begin{equation}
\mathcal{M}_{\theta,{\rm ext}}\cap\left\{  y_{0}=b\right\}  =\mathcal{M}_{\theta
}\left(  b\right)  \label{T7}%
\end{equation}
for any $b\in\left(  0,\frac{1}{2}\right)  .$ Indeed, the manifold
$\mathcal{M}_{\theta,{\rm ext}}$ is any centre-stable manifold at the point 
$\left(
y_{0},Q_{1},Q_{2},G,Z\right)  =\left(  0,0,0,\frac{1}{3},0\right)  $
associated to the system (\ref{F1E1})-(\ref{F1E4}) complemented with the
additional equation
\begin{equation}
\frac{dy_{0}}{d\zeta}=0.\label{T4a}%
\end{equation}
More precisely, we make use of the fact that the dynamical system of interest
has a smooth extension to an open neighbourhood of the stationary point
under consideration. The manifold $\mathcal{M}_{\theta,{\rm ext}}$ is
the intersection of a centre-stable manifold for the extended system with the 
subset defined by the inequality $y_0\ge 0$.
The manifold $\mathcal{M}_{{\rm ext}}$ contains all the points of the form 
$\left(
y_{0},P_{1}\left(  y_{0}\right)  \right)  $ with $y_{0}\in\left[  0,\frac
{1}{2}\right]  $ since they remain in a neighbourhood of $\left(
0,0,0,\frac{1}{3},0\right)  $ for arbitrary times. Moreover, the manifolds
$\mathcal{M}_{\theta,{\rm ext}}\cap\left\{  y_{0}=b\right\}  $ are invariant under
the flow (\ref{F1E1})-(\ref{F1E4}) and since they are formed by points that
remain in a neighbourhood of $\left(  0,0,0,\frac{1}{3},0\right)  $ for
arbitrarily long times, it follows from (\ref{T4a}) that the points in
$\mathcal{M}_{\theta,{\rm ext}}\cap\left\{  y_{0}=b\right\}  $ are contained in 
the
stable manifold associated to the point $P_{1}\left(  y_{0}\right)  .$ The
uniqueness of the stable manifold then implies $\mathcal{M}_{\theta}\left(
b\right)  \subset\mathcal{M}_{\theta,{\rm ext}}\cap\left\{  y_{0}=b\right\}  .$
Moreover, the form of the tangent space to $\mathcal{M}_{\theta,{\rm ext}}$ at 
the
point $\left(  0,0,0,\frac{1}{3},0\right)  $ implies that the dimension of
$\mathcal{M}_{\theta,{\rm ext}}\cap\left\{  y_{0}=b\right\}  $ is three for small
$b.$ Since this is also the dimension of $\mathcal{M}_{\theta}\left(
b\right)  $ the relation (\ref{T7}) follows. The continuity of $\mathcal{M}%
_{{\rm ext}}$ then implies that the centre-stable manifold $\mathcal{M}_{\theta
}\left(  0\right)  $ can be uniquely obtained as limit of the manifolds
$\mathcal{M}_{\theta}\left(  y_{0}\right)  $ as $y_{0}\rightarrow0^{+}.$ In
particular the manifold $\mathcal{M}_{\theta}\left(  0\right)  $ is unique.

The properties of the manifold $\mathcal{M}_{\theta}\left(  0\right)  $ can be
analysed in more detail. We remark that the curve:%

\begin{equation}
\sqrt{\left(  Z^{2}+1\right)  }\sqrt{G}\left(  1-G\right)  =\frac{2}%
{3^{\frac{3}{2}}}\ \ ,\ \ Q_{1}=Q_{2}=0 \label{T6}%
\end{equation}
belongs to $\mathcal{M}_{\theta}\left(  0\right)  $ since the hyperplane
$\left\{  Q_{1}=Q_{2}=0\right\}  $ is invariant under the dynamics induced by
(\ref{T1})-(\ref{T4}). On the other hand, the invariance of (\ref{T1}%
)-(\ref{T4}) under rotations in the $\left(  Q_{1},Q_{2}\right)  $-plane
allows the problem to be reduced to one with smaller dimensionality. More
precisely, defining $Q=\sqrt{\frac12 (Q_{1}^{2}+Q_{2}^{2})}$ leads to the
system:
\begin{align}
\frac{dQ}{d\zeta}  &  =-2GZQ,\label{Q5a}\\
\frac{dG}{d\zeta}  &  =2G\left[  Z\left(  1-G\right)  -\theta\left[
Z^{2}+1\right]  ^{\frac{3}{2}}Q^{2}\right]  ,\label{Q5b}\\
\frac{dZ}{d\zeta}  &  =\left(  3G-1-\theta ZQ^{2}\sqrt{\left(  Z^{2}+1\right)
}\right)  \left(  Z^{2}+1\right)  . \label{Q5c}%
\end{align}
We will denote by $\mathcal{N}_{\theta}$ the (two-dimensional) invariant
manifold associated to the system (\ref{Q5a})-(\ref{Q5c}) that is obtained
from ${\mathcal{M}}_{\theta}$ by taking the quotient by rotations in the
$Q_{i}$ and which contains the curve (\ref{T6}).

Our goal is to show the existence for any $y_{0}$ sufficiently small of a
value $\theta^{\ast}=\theta^{\ast}\left(  y_{0}\right)  $ of $\theta$ such
that the manifold $\mathcal{M}_{\theta^{\ast}}\left(  y_{0}\right)  $ contains
the point $Q_{1}=Q_{2}=1,\ G=1,\ Z=0.$ This will be done by showing that the
corresponding statement holds in the case $y_{0}=0$ and then doing a
perturbation argument. The statement about the manifold $\mathcal{M}%
_{\theta^{\ast}}\left(  0\right)  $ is equivalent to the statement that
${\mathcal{N}}_{\theta^{\ast}}$ contains the point $(1,1,0)$. It will be shown
that the latter statement is true and, moreover, that when $\theta$ is varied
through the value $\theta^{*}$ the manifold ${\mathcal{N}}_{\theta}$ moves
through $(1,1,0)$ with non-zero velocity. It then follows that $\mathcal{M}%
_{\theta}(0)$ moves through $(1,1,1,0)$ with non-zero velocity. Note that 
the coefficients of the system extend smoothly to an open neighbourhood of 
the manifold $\mathcal{M}_{\theta^{\ast}}\left(  0\right)$. As a consequence 
the manifold  $\mathcal{M}_{\theta,{\rm {\rm ext}}}$ extends smoothly to small
negative values of $y_0$. The desired
statement concerning $\mathcal{M}_{\theta}(y_{0})$ is a consequence of the
implicit function theorem. In more detail, the statement that $\mathcal{M}%
_{\theta}$ depends on $\theta$ and $y_{0}$ in a way which is continuously
differentiable means that there is a $C^{1}$ mapping $\Psi$ from the product
of a neighbourhood of $(0,\theta^{*})$ in ${\mathbb{R}}^{2}$ with
$\mathcal{M}_{\theta}(0)$ into a neighbourhood of $(1,1,1,0)$ with the
properties that its restriction to $y_{0}=0$ and $\theta=\theta^{*}$ is the
identity and that the image of $\{(y_{0},\theta)\}\times\mathcal{M}%
_{\theta^{*}}(0)$ under $\Psi$ is $\mathcal{M}_{\theta}(y_{0})$. The condition
that the manifold moves with non-zero velocity implies that if $x_{0}$ denotes
the point of ${\mathcal{M}}_{\theta^{*}}(0)$ with coordinates $(1,1,1,0)$ the
linearization of $\Psi$ at the point $(0,\theta^{*},x_{0})$ with respect to
the last four variables is an isomorphism. This allows the implicit function
theorem to be applied.

In order to check the existence of $\theta^{*}$ it is enough to study the
behaviour of the manifolds $\mathcal{N}_{\theta}$ for $\theta\rightarrow0^{+}$
and $\theta\rightarrow\infty.$ These manifolds are two-dimensional manifolds
in the three-dimensional space $\left(  Q,G,Z\right)  .$ Notice that the
structure of the manifolds $\mathcal{N}_{\theta}$ can be easily understood
using the fact that the parameter $\theta$ can be rescaled out of the system
(\ref{Q5a})-(\ref{Q5c}) using the change of variables:%

\begin{equation}
Q=\frac{1}{\sqrt{\theta}}q. \label{R1E1}%
\end{equation}
Then (\ref{Q5a})-(\ref{Q5c}) becomes:%
\begin{align}
\frac{dq}{d\zeta}  &  =-2GZq,\label{R1E2}\\
\frac{dG}{d\zeta}  &  =2G\left[  Z\left(  1-G\right)  -\left[  Z^{2}+1\right]
^{\frac{3}{2}}q^{2}\right]  ,\label{R1E3}\\
\frac{dZ}{d\zeta}  &  =\left(  3G-1-Zq^{2}\sqrt{\left(  Z^{2}+1\right)
}\right)  \left(  Z^{2}+1\right)  . \label{R1E4}%
\end{align}

Let us denote by $\widetilde{\mathcal{N}}$ the centre-stable manifold at the
point $\left(  q,G,Z\right)  =\left(  0,\frac{1}{3},0\right)  $ for the
dynamics (\ref{R1E2})-(\ref{R1E4}). The manifold $\widetilde{\mathcal{N}}$
contains the curve $\left\{  \left(  Z^{2}+1\right)  G\left(  1-G\right)
^{2}=\frac{4}{3^{3}},\;q=0\right\}  .$ Notice that:%
\[
\left(  Q,G,Z\right)  \in\mathcal{N}_{\theta}\Longleftrightarrow\left(
\sqrt{\theta}Q,G,Z\right)  \in\widetilde{\mathcal{N}}.
\]
Therefore the family of manifolds $\mathcal{N}_{\theta}$ can be obtained from
the manifold $\widetilde{\mathcal{N}}$ by means of the rescaling (\ref{R1E1})
while keeping the same value of the variables $G,\ Z.$ In order to check if
$\left(  Q,G,Z\right)  =\left(  1,1,0\right)  \in{\mathcal{N}}_{\theta}$ we
just need to describe in detail the intersection of the manifold
$\widetilde{\mathcal{N}}$ with the line $\left\{  G=1,Z=0\right\}  .$ Once the
existence of a value $\theta^{*}$ of $\theta$ for which the manifold
${\mathcal{N}}_{\theta^{*}}$ contains the point $(1,1,0)$ has been shown the
statement that the manifold ${\mathcal{N}}_{\theta}$ moves through this point
with non-zero velocity follows immediately from the rescaling property.

Notice that the plane $\left\{  q=0\right\}  $ is invariant for the system of
equations (\ref{R1E2})-(\ref{R1E4}). The analysis of the trajectories of
(\ref{R1E2})-(\ref{R1E4}) in this plane can be done using phase portrait
arguments. There is a unique equilibrium point at $\left(  G,Z\right)
=\left(  \frac{1}{3},0\right)  $ with stable manifold $\left\{  \left(
Z^{2}+1\right)  G\left(  1-G\right)  ^{2}=\frac{4}{3^{3}}\right\}  .$ This
manifold splits the plane $\left(  G,Z\right)  $ in two connected regions. The
trajectories starting their motion in the region that contains the point
$\left(  G,Z\right)  =\left(  0,0\right)  $ reach the line $Z=0$ for a finite
value of $\zeta$ if $Z>0$ initially and eventually develop a singularity where
$Z$ approaches $-\infty$ at a finite value of $\zeta.$ On the other hand, the
trajectories starting their motion in the region containing the point $\left(
G,Z\right)  =\left(  1,0\right)  $ move in the direction of increasing $Z$
towards $Z=\infty,\ G=1,$ a value that is achieved for a finite value of
$\zeta.$

Notice that the solutions of (\ref{R1E2})-(\ref{R1E4}) starting their dynamics
in the set $\left\{  0\leq G\leq1,\ Z\geq0\right\}  $ can only evolve in two
different ways. Either the trajectory remains in the region where $Z\geq0$ for
arbitrarily large values of $\zeta$ or the trajectory enters the region
$\left\{  Z<0\right\}  .$ In the second case this can only happen through the
set $G\leq\frac{1}{3}.$ Since $G$ is decreasing it remains in the set
$\left\{  Z<0\right\}  $ for larger values of $\zeta$ and eventually it
approaches $Z=-\infty$ for some finite value of $\zeta.$

Suppose otherwise that the trajectory remains in the region where $Z\geq0$ for
arbitrary values of $\zeta.$ Then $q$ decreases to zero and the behaviour of
the trajectories is then similar to the ones in the plane $\left\{
q=0\right\}  .$ We now claim that either this trajectory belongs to the stable
manifold $\widetilde{\mathcal{N}}$ or it satisfies $\lim_{\zeta\rightarrow
\zeta^{\ast}}Z\left(  \zeta\right)  =\infty$ for some $\zeta^{\ast}\leq
\infty.$ In order to avoid breaking the continuity of the argument we will
prove this result in Lemma \ref{LeT1} in Section \ref{TechLemm}.

We will show that there exists a point of the line $\left\{  G=1,Z=0\right\}
$ in the manifold $\widetilde{\mathcal{N}}.$ The points of this line enter the
region $\left\{  0<G<1,\ Z>0\right\}  $ due to the form of the vector field
associated to (\ref{R1E2})-(\ref{R1E4}). If $q\left(  0\right)  >0$ is small,
Lemma \ref{LeT2} shows that $Z$ approaches $Z=\infty$ for a finite value of
$\zeta.$ Suppose now that $q\left(  0\right)  >0$ is sufficiently large. Then
the trajectory enters the region $\left\{  Z<0\right\}  $ for a finite value
of $\zeta$ as the following argument shows. A solution which starts at
$(q_{0},1,0)$ with $q_{0}$ large immediately enters the region $Z>0$, $G<1$.
The inequality $Z\le1$ will hold for at least a time $\frac14$ since
$\frac{dZ}{d\zeta}\le4$ as long as $Z\le1$. The aim is to show that for
$q_{0}$ sufficiently large $Z$ will become negative within the interval
$[0,\frac14]$. From now on only that interval is considered. Integrating the
equation for $q$ gives the inequality $q(\zeta)\ge e^{-\frac12}q_{0}$. The
equation for $G$ then shows that $G(\zeta)\le e^{-\alpha(q_{0})\zeta}$ where
$\alpha(q_{0})=q_{0}^{2}e^{-1}-1$. Choose $q_{0}$ large enough so that
$e^{-\frac1{40} \alpha(q_{0})}\le\frac16$. When $\zeta=\frac1{40}$ the
inequality $Z\le\frac1{10}$ still holds. Under the given circumstances $G$ is
decreasing on the whole interval $[0,\frac14]$. The equation for $Z$ shows
that by the time $\zeta=\frac9{40}$ at the latest $Z$ has reached zero.

Let $U_{1}$ be the set of positive real numbers $q_{0}$ for which the solution
starting at $(q_{0},1,0)$ is such that $Z\rightarrow-\infty$ as $\zeta
\rightarrow\zeta^{\ast}$, where $\zeta^{\ast}$ denotes the maximal time of
existence, and let $U_{2}$ be the set of positive real numbers $q_{0}$ for
which the solution starting at $(q_{0},1,0)$ is such that $Z\rightarrow
+\infty$ as $\zeta\rightarrow\zeta^{\ast}$. It follows from Lemma \ref{LeT3}
that $U_{2}$ is open. We also know that $U_{1}$ is open. Moreover, it has been 
proved that both $U_1$ and $U_2$ are non-empty. By connectedness of
the interval $(0,\infty)$ it follows that there must be a value of $q_{0}$ for
which the solution starting at $(q_{0},1,0)$ is neither in $U_1$ or $U_2$.
For that solution $Z$ is non-negative and does not tend to infinity and thus, 
by Lemma \ref{LeT1}, it is the desired solution which lies on 
$\tilde{\mathcal{N}}$. 

The equivalence between the existence of the self-similar solution described
in Section \ref{SelfSimD} and the existence of a trajectory connecting the
points $\left(  Q_{1},Q_{2},G,Z\right)  =\left(  1,1,1,0\right)  $ and
$\left(  Q_{1,\infty},Q_{2,\infty},\Lambda_{\infty},u_{\infty}\right)  $
proved in Subsection \ref{refTraj} concludes the proof of Theorem
\ref{shootingTheorem}. Theorem \ref{Th1} is just a Corollary of Theorem
\ref{shootingTheorem}.
\end{proof}

\bigskip

\section{\label{TechLemm}Some auxiliary lemmas used in the analysis of
(\ref{R1E2})-(\ref{R1E4}).}

\bigskip

\begin{lemma}
\label{LeT1} Suppose that a solution of (\ref{R1E2})-(\ref{R1E4}) is defined
for $\zeta\in\left[  \zeta_{\ast},\zeta^{\ast}\right)  ,$ where $\zeta^{\ast}$
is the maximal time of existence. Suppose that $Z\left(  \zeta\right)  >0$ for
any $\zeta\in\left[  \zeta_{\ast},\zeta^{\ast}\right)  $ and also that
$G\left(  \zeta_{\ast}\right)  \in\left(  0,1\right)  ,\ q\left(  \zeta_{\ast
}\right)  >0.$ Then, either the curve $\left\{  \left(  q\left(  \zeta\right)
,G\left(  \zeta\right)  ,Z\left(  \zeta\right)  \right)  :\zeta\in\left(
\zeta_{\ast},\zeta^{\ast}\right)  \right\}  $ is contained in the stable
manifold $\widetilde{\mathcal{N}}$ or $\lim_{\zeta\rightarrow\zeta^{\ast}%
}Z\left(  \zeta\right)  =\infty.$
\end{lemma}

\noindent

\begin{proof}
The plane $\left\{  G=0\right\}  $ is invariant under the flow associated to
(\ref{R1E2})-(\ref{R1E4}). On the other hand, the vector field on the
right-hand side of (\ref{R1E2})-(\ref{R1E4}) points into the region $\left\{
G<1\right\}  $ if $q\neq0.$ Therefore the region $\left\{
0<G<1,\ q>0\right\}  $ is invariant for the flow defined by (\ref{R1E2}%
)-(\ref{R1E4}) and we can assume that the inequalities $0<G\left(
\zeta\right)  <1,\ q\left(  \zeta\right)  >0$ hold for any $\zeta\in\left[
\zeta_{\ast},\zeta^{\ast}\right)  .$ We now have two possibilities:%
\begin{align}
\limsup_{\zeta\rightarrow\zeta^{\ast}}Z\left(  \zeta\right)   &
<\infty\text{ },\label{Alt1}\\
\limsup_{\zeta\rightarrow\zeta^{\ast}}Z\left(  \zeta\right)   &  =\infty.
\label{Alt2}%
\end{align}
Suppose first that (\ref{Alt1}) holds. Then, there exists $M>0$ such that
\begin{equation}
Z\left(  \zeta\right)  \leq M\ \ \ \text{for any \ }\zeta\in\left[
\zeta_{\ast},\zeta^{\ast}\right)  . \label{BoundZ}%
\end{equation}
We claim that in this case the trajectory $\left\{  \left(  q\left(
\zeta\right)  ,G\left(  \zeta\right)  ,Z\left(  \zeta\right)  \right)
:\zeta\in\left(  \zeta_{\ast},\zeta^{\ast}\right)  \right\}  $ is contained in
$\widetilde{\mathcal{N}}$. Notice that in this case, the boundedness of
$\left\vert \left(  q,G,Z\right)  \right\vert $ implies that $\zeta^{\ast
}=\infty.$ Since $\left(  GZq\right)  \left(  \zeta\right)  >0$ for $\zeta
\in\left[  \zeta_{\ast},\infty\right)  $ it follows from (\ref{R1E2}) that
$q\left(  \zeta\right)  $ is decreasing. Therefore
$q_{\infty}=\lim_{\zeta\rightarrow\infty}q\left(  \zeta\right)$ exists and is
non-negative. Suppose
that $0<q_{\infty}.$ Then $0<q_{\infty}<q\left(  \zeta\right)  $ for any
$\zeta\in\left[  \zeta_{\ast},\infty\right)  .$ Integrating (\ref{R1E2}) we
obtain $\int_{\zeta_{\ast}}^{\infty}\left(  GZq\right)  \left(  \zeta\right)
d\zeta<\infty,$ whence
\begin{equation}
\int_{\zeta_{\ast}}^{\infty}\left(  GZ\right)  \left(  \zeta\right)
d\zeta<\infty. \label{Alt3}%
\end{equation}
Since $\frac{dG}{d\zeta},\ \frac{dZ}{d\zeta}$ are bounded, (\ref{Alt3})
implies $\lim_{\zeta\rightarrow\infty}\left(  GZ\right)  \left(  \zeta\right)
=0.$ Then (\ref{R1E3}) implies:%
\[
\frac{dG}{d\zeta}\leq-q_{\infty}^{2}G
\]
for $\zeta\geq\zeta_{0}$ sufficiently large. Therefore $\lim_{\zeta
\rightarrow\infty}G\left(  \zeta\right)  =0.$ Equation (\ref{R1E4}) then
yields:
\[
\frac{dZ}{d\zeta}\leq-\frac{1}{2}%
\]
for $\zeta\geq\zeta_{0}$ large enough. Then $Z\left(  \zeta\right)  <0$ for
large $\zeta,$ but this contradicts the hypothesis of the lemma. It then
follows that $q_{\infty}=0.$

Due to (\ref{BoundZ}) and since $\lim_{\zeta\rightarrow\infty}q\left(
\zeta\right)  =0$ we can approximate the trajectories associated to
(\ref{R1E2})-(\ref{R1E4}) for large values of $\zeta$ using the corresponding
trajectories associated to (\ref{R1E2})-(\ref{R1E4}) for $q=0.$ The study of
the trajectories associated to (\ref{R1E2})-(\ref{R1E4}) that are contained in
$\left\{  q=0\right\}  \cap\left\{  0<G<1\right\}  $ reduces to a
two-dimensional phase portrait. These trajectories can have only three
different behaviours. Either they are contained in $\widetilde{\mathcal{N}%
}\cap\left\{  q=0\right\}  ,$ or they reach the plane $\left\{  Z=0\right\}
,$ with $G<\frac{1}{3},$ entering $\left\{  Z<0\right\}  ,$ or they become
unbounded. The continuous dependence of the trajectories with respect to the
initial values as well as the fact that $\lim_{\zeta\rightarrow\infty}q\left(
\zeta\right)  =0$ implies then that either $\lim_{\zeta\rightarrow\infty
}\mathrm{dist}\left(  \left(  q\left(  \zeta\right)  ,G\left(  \zeta\right)
,Z\left(  \zeta\right)  \right)  ,\widetilde{\mathcal{N}}\cap\left\{
q=0\right\}  \right)  =0,$ or $Z\left(  \zeta\right)  <0$ for some
$\zeta<\infty$, or $Z\left(  \zeta\right)  \geq M+1$ for some $\zeta<\infty.$
The second alternative contradicts the hypothesis of the lemma. The third
alternative contradicts (\ref{BoundZ}) and therefore only the first
alternative is left. However, in that case $\lim_{\zeta\rightarrow\infty}$
$\left(  q\left(  \zeta\right)  ,G\left(  \zeta\right)  ,Z\left(
\zeta\right)  \right)  =\left(  0,\frac{1}{3},0\right)  $ and the trajectory
is contained in $\widetilde{\mathcal{N}}$ as claimed.

Suppose then that (\ref{Alt2}) holds. We claim that in this case $\lim
_{\zeta\rightarrow\zeta^{\ast}}Z\left(  \zeta\right)  =\infty.$ Notice that
the monotonicity of $q\left(  \zeta\right)  $ implies that $\lim
_{\zeta\rightarrow\zeta^{\ast}}q\left(  \zeta\right)  =q_{\infty}$ exists. We
will first prove that $q_{\infty}=0.$ Suppose that, on the contrary,
$q_{\infty}>0.$ Then $q\left(  \zeta\right)  >q_{\infty}>0$ for any $\zeta
\in\left[  \zeta_{\ast},\zeta^*\right)  .$ Equation (\ref{R1E4}) as well as
$G<1$ yields:%
\[
\frac{dZ}{d\zeta}<\left(  2-Zq_{\infty}^{2}\sqrt{Z^{2}+1}\right)  \left(
Z^{2}+1\right)
\]
for any $\zeta\in\left[  \zeta_{\ast},\zeta^*\right)  .$ This inequality
implies $\frac{dZ}{d\zeta}<0$ for $Z>Z_{\infty}=Z_{\infty}\left(  q_{\infty
}\right)  .$ Therefore $Z\left(  \zeta\right)  <Z_{\infty}$ for $\zeta
\in\left[  \zeta_{\ast},\zeta^*\right)  $ and this contradicts (\ref{Alt2}).
From now on take $q_{\infty}=0.$ We can then assume (\ref{Alt2}) and
\begin{equation}
\lim_{\zeta\rightarrow\zeta^{\ast}}q\left(  \zeta\right)  =0.\label{limQ}%
\end{equation}
Suppose also that $\liminf_{\zeta\rightarrow\zeta^{\ast}}Z\left(
\zeta\right)  <\infty.$ This is equivalent to the existence of $0<M<\infty$
and a subsequence \ $\left\{  \zeta_{n}\right\}  $ with $\zeta_{n}%
\rightarrow\zeta^{\ast}$ as $n\rightarrow\infty$ such that:%
\begin{equation}
Z\left(  \zeta_{n}\right)  \leq M.\label{In1}%
\end{equation}
We now claim that:%
\begin{equation}
\lim_{\zeta\rightarrow\zeta^{\ast}}\left[  Z\left(  \zeta\right)  q\left(
\zeta\right)  \right]  =0.\label{E3}%
\end{equation}
To prove (\ref{E3}) we argue as follows. Combining (\ref{R1E2}), (\ref{R1E4})
we obtain:
\begin{equation}
\frac{d}{d\zeta}\left(  Zq\right)  =qZ^{2}\left(  G-1\right)  +q\left(
3G-1\right)  -Zq^{3}\sqrt{\left(  Z^{2}+1\right)  ^{3}}.\label{dZq}%
\end{equation}
We now use the inequality $Z\sqrt{\left(  Z^{2}+1\right)  ^{3}}\geq Z^{4}$ for
$Z>0.$ Then, using also the inequality $G<1:$%
\begin{equation}
\frac{d}{d\zeta}\left(  Zq\right)  \leq q^{-1}\left[  \left(  3G-1\right)
q^{2}-\left(  Zq\right)  ^{4}\right]  .\label{dZq1}%
\end{equation}
It follows from this inequality, as well as (\ref{limQ}) that for any
$\varepsilon>0,$ every trajectory satisfying the hypothesis of Lemma
\ref{LeT1} and entering any of the regions $\left\{  \left(  q,G,Z\right)
:Zq<\varepsilon\text{ }\right\}  $ for $\zeta$ sufficiently close to
$\zeta^{\ast}$ remains in such a region for later times. If $\zeta^{\ast
}=\infty,$ the meaning of sufficiently close is large enough. Due to
(\ref{limQ}) and (\ref{In1}), for any $\varepsilon>0,$ there exist $\zeta_{n}$
arbitrarily close to $\zeta^{\ast}$ such that $\left(  Zq\right)  \left(
\zeta_{n}\right)  <\varepsilon.$ Then $\left(  Zq\right)  \left(
\zeta\right)  <\varepsilon$ for any $\zeta\in\left(  \zeta_{n},\zeta^{\ast
}\right)  .$ Since $\varepsilon$ is arbitrary we obtain (\ref{E3}).

Combining (\ref{limQ}) and (\ref{E3}) it follows that:%
\begin{equation}
\lim_{\zeta\rightarrow\zeta^{\ast}}\delta_{1}\left(  \zeta\right)
=\lim_{\zeta\rightarrow\zeta^{\ast}}\delta_{2}\left(  \zeta\right)  =0
\label{E6}%
\end{equation}
where:%
\[
\delta_{1}\left(  \zeta\right)  =Zq^{2}\sqrt{Z^{2}+1}\ \ \ ,\ \ \ \delta
_{2}\left(  \zeta\right)  =\frac{\left(  Z^{2}+1\right)  ^{\frac{3}{2}}q^{2}%
}{Z+1}.
\]
We can then rewrite (\ref{R1E2}), (\ref{R1E4}) as:
\begin{align}
\frac{dq}{d\zeta}  &  =-2GZq,\label{EDO1}\\
\frac{dG}{d\zeta}  &  =2G\left[  Z\left(  1-G\right)  -\left(  Z+1\right)
\delta_{2}\left(  \zeta\right)  \right]  ,\label{ED02}\\
\frac{dZ}{d\zeta}  &  =\left(  3G-1-\delta_{1}\left(  \zeta\right)  \right)
\left(  Z^{2}+1\right)  . \label{ED03}%
\end{align}
We now claim the following. Given any $\varepsilon_{0}$ belonging to the
interval $\left(  0,\frac23\right)  $ suppose that the trajectory under
consideration enters the set:%
\[
\Omega_{\varepsilon_{0}}=\left\{  G\geq\frac{1}{3}+\varepsilon_{0}%
\ ,\ Z\geq1\right\}
\]
for some $\zeta<\zeta^{\ast}$ sufficiently close to $\zeta^{\ast}.$ Then
$\lim_{\zeta\rightarrow\zeta^{\ast}}Z\left(  \zeta\right)  =\infty$ and
$\zeta^{\ast}<\infty.$ The proof as the follows. Due to (\ref{E6}) the set
$\Omega_{\varepsilon_{0}}$ is invariant for (\ref{EDO1})-(\ref{ED03}) if
$\zeta$ is close to $\zeta^{\ast}.$ Then, for $\zeta$ close to $\zeta^{\ast}$
we have:%
\[
\frac{dZ}{d\zeta}\geq\varepsilon_{0}\left(  Z^{2}+1\right)
\]
and this implies $\lim_{\zeta\rightarrow\zeta^{\ast}}Z\left(  \zeta\right)
=\infty$ and $\zeta^{\ast}<\infty.$

Therefore, to complete the proof of Lemma \ref{LeT1} it only remains to prove
that the trajectory enters $\Omega_{\varepsilon_{0}}$ for values of $\zeta$
sufficiently close to $\zeta^{\ast}.$ Due to (\ref{Alt2}) and (\ref{In1})
there exists a sequence $\left\{  \bar{\zeta}_{n}\right\}  $ with $\zeta
_{n}<\bar{\zeta}_{n}<\zeta^{\ast}$ such that:%
\[
Z\left(  \bar{\zeta}_{n}\right)  =2M\ \ \ \text{and\ \ \ }\frac{dZ}{d\zeta
}\left(  \bar{\zeta}_{n}\right)  \geq0.
\]
Due to (\ref{ED03}) this implies:%
\begin{equation}
\limsup_{n\rightarrow\infty}G\left(  \bar{\zeta}_{n}\right)  \geq\frac{1}{3}.
\label{E8}%
\end{equation}
On the other hand, a Gronwall type of argument applied to (\ref{ED03}) implies
the existence of $\alpha_{M}>0$, depending only on $M$ such that:%
\begin{equation}
0<\frac{M}{2}\leq Z\left(  \zeta\right)  \leq4M\ \ \ \text{for\ \ \ }\zeta
\in\left[  \bar{\zeta}_{n},\bar{\zeta}_{n}+\alpha_{M}\right]  . \label{E9}%
\end{equation}
Comparing the solution of the equation (\ref{ED02}) with the solution of the
equation $\frac{dG}{d\zeta}=2GZ\left(  1-G\right)  $ with the same initial
datum at $\zeta=\bar{\zeta}_n$ and taking into account (\ref{E8}), (\ref{E9}) it
then follows that, for $n$ large enough $\left(  q\left(  \bar{\zeta}%
_{n}\right)  ,G\left(  \bar{\zeta}_{n}\right)  ,Z\left(  \bar{\zeta}%
_{n}\right)  \right)  \in\Omega_{\varepsilon_{0}}.$ Therefore $\lim
_{\zeta\rightarrow\zeta^{\ast}}Z\left(  \zeta\right)  =\infty.$ This
contradicts (\ref{In1}) and the lemma follows.
\end{proof}

\begin{lemma}
\label{LeT2} There exists $\delta>0$ sufficiently small such that, the
solution of (\ref{R1E2})-(\ref{R1E4}) with initial value $\left(  q\left(
0\right)  ,G\left(  0\right)  ,Z\left(  0\right)  \right)  =\left(
q_{0},1,0\right)  $ and $0<q_{0}<\delta$ satisfies $\lim_{\zeta\rightarrow
\zeta^{\ast}}Z\left(  \zeta\right)  =\infty,$ where $\zeta^{\ast}$ denotes the
maximal time of existence of the trajectory.
\end{lemma}

\noindent

\begin{proof}
The trajectory enters the region $\left\{  Z>0\right\}  $ and as long as it
remains there, the function $q\left(  \zeta\right)  $ is decreasing. The
inequality $\frac{\partial Z}{\partial\zeta}\le 4$ holds as long as $Z\le 1$.
It follows that $Z\le 1$ on the interval $\left[0,\frac14\right]$. On that
interval the inequality 
$\frac{\partial (\log G)}{\partial\zeta}\ge -2^{\frac52}q_0^2$ holds and
hence $G\ge e^{-q_0^2}$. Furthermore 
\begin{equation}
\frac{\partial Z}{\partial\zeta}\ge 3e^{-q_0^2}-1-\sqrt{2}q_0^2
=\beta(q_0).   
\end{equation}
Choose $\delta$ sufficiently small that $\beta(\delta)> 1$ and 
$e^{-\delta^2}>\frac12$. Then $Z\left(\frac14\right)>\frac14$ and 
$G>\frac12$ on $\left[0,\frac14\right]$. Choose $\epsilon>0$ and suppose
that $2\delta^2<\epsilon^4$. Then it follows from (\ref{dZq1}) that the set 
defined by the inequality $Zq\le\epsilon$ is invariant. Thus the solution
remains in that region on its whole interval of existence. Now
$\delta_1(\zeta)\le\epsilon \sqrt{\epsilon^2+\delta^2}$ and
$\delta_2(\zeta)\le (\epsilon^2+\delta^2)$. Let $[0,\zeta_1)$ be the longest 
interval on which $G\ge\frac12$. From what has been shown already
$\zeta_1\ge \frac14$. Reduce the size of $\epsilon$ if necessary so that
$\epsilon\sqrt{\epsilon^2+\delta^2}<\frac12$. Then it follows from 
(\ref{ED03}) that $Z$ is increasing on $[0,\zeta_1)$ and hence is greater than
$\frac14$ for $\zeta\ge\zeta_1$. Putting this information into
(\ref{ED02}) shows that provided $\epsilon^2+q_0^2<\frac1{16}$ then 
$G$ cannot decrease. For $\delta$ sufficiently small this gives a 
contradiction unless $\zeta_1=\zeta^*$. In particular there is a positive 
lower bound for $Z$ at late times. Furthermore (\ref{ED03}) implies that
$\lim_{\zeta\rightarrow\zeta^{\ast}}Z\left(  \zeta\right)  =\infty$ and the 
lemma follows.
\end{proof}

\begin{lemma}
\label{LeT3} Suppose that a solution satisfying the hypotheses of Lemma
\ref{LeT1} with $\zeta_*=0$ has the property that 
$\lim_{\zeta\to\zeta^{*}}Z(\zeta)=\infty$. Then any solution starting 
sufficiently close to the given solution for $\zeta=0$ also has the property 
that $Z$ tends to infinity on its maximal interval of existence.
\end{lemma}

\begin{proof}
To start with a number of further consequences of the hypotheses of Lemma
\ref{LeT1} will be derived. The assumption on the initial condition only
plays a role towards the end of the proof.
It has been shown in the proof of Lemma \ref{LeT1} that $\lim_{\zeta
\rightarrow\zeta^{\ast}}q\left(  \zeta\right)  =0$. We now claim that
(\ref{E3}) holds. Suppose that it is not true. Then we claim that the limit
$\lim_{\zeta\rightarrow\zeta^{\ast}}\left(  Zq\right)  \left(  \zeta\right)
=L$ exists and that $L>0$. Indeed, notice first that $\liminf_{\zeta
\rightarrow\zeta^{\ast}}\left(  Zq\right)  \left(  \zeta\right)  >0.$
Otherwise there would exist a sequence $\left\{  \zeta_{n}\right\}  $ such
that $\lim_{n\rightarrow\infty}\zeta_{n}=\zeta^{\ast}$ with $\lim
_{n\rightarrow\infty}\left(  Zq\right)  \left(  \zeta_{n}\right)  =0.$
Combining this with the fact that $q\rightarrow0$ and (\ref{dZq1}) we would
obtain (\ref{E3}), a contradiction. Thus $\liminf_{\zeta\rightarrow\zeta^{\ast
}}\left(  Zq\right)  \left(  \zeta\right)  >0.$ Using again the fact that
$q\rightarrow0$ and (\ref{dZq1}) it follows that $\left(  Zq\right)  $ is
monotone decreasing for $\zeta$ close to $\zeta^{\ast},$ whence the limit
$\lim_{\zeta\rightarrow\zeta^{\ast}}\left(  Zq\right)  \left(  \zeta\right)
=L$ exists. Moreover we have obtained also in this case that $\left(
Zq\right)  \left(  \zeta\right)  >L$ for $\zeta$ close to $\zeta^{\ast}.$

It follows from the proof of Lemma \ref{LeT1} that $\zeta^*<\infty$. 
By the boundedness of the right hand side of (\ref{R1E2}) it follows by
integrating this equation between $\zeta$ and $\zeta^*$ that
$q(\zeta)\le a^{-1}(\zeta^*-\zeta)$ for a positive constant $a$. Hence
$q^{-1}(\zeta)\ge a(\zeta^*-\zeta)^{-1}$. This can be used together with the
limiting behaviour of $Zq$ to estimate the right hand side of (\ref{dZq})
from above. The first term is negative and can be discarded. The second
term tends to zero as $\zeta\to\zeta^*$. The third term can be written in a 
suggestive form as $-q^{-1}[(Zq)\sqrt{((Zq)^2+q^2)^3}]$. The expression in 
square 
brackets tends to a positive limit as $\zeta\to\zeta^*$. Thus the right hand
side of (\ref{dZq}) fails to be integrable, contradicting the fact that
$Zq$ is positive.
This contradiction
completes the proof that 
$\lim_{\zeta\rightarrow\zeta^{\ast}}\left(  Zq\right)  \left(
\zeta\right)  =0$.

We now use some arguments analogous to the ones used in the proof of Lemmas
\ref{LeT1} and \ref{LeT2}. As a next step we prove that
$G(\zeta)$ tends to a limit as $\zeta\to\zeta^*$ and that this limit is
greater than $\frac13$. We first claim that:%
\begin{equation}
S=\limsup_{\zeta\rightarrow\zeta^{\ast}}G\left(  \zeta\right)  \geq\frac
{1}{3}. \label{B9}%
\end{equation}
Indeed, suppose first that $S=\limsup_{\zeta\rightarrow\zeta^{\ast}}G\left(
\zeta\right)  <\frac{1}{3}.$ Since $\lim_{\zeta\rightarrow\zeta^{\ast}}\left(
Zq\right)  \left(  \zeta\right)  =0$ we can approximate (\ref{Q5a}%
)-(\ref{Q5c}) by the system (\ref{EDO1})-(\ref{ED03}). Using (\ref{ED03}) 
it follows that $Z\left(  \zeta\right)  $ is
decreasing for $\zeta$ close to $\zeta^{\ast}.$ This contradicts (\ref{Alt2})
and then (\ref{B9}) follows.
On the other hand (\ref{ED02}) implies that $G$ is increasing if $G>\frac
{1}{4}$ for $\zeta$ close to $\zeta^{\ast}.$ Using (\ref{B9}) it then follows
that $G$ increases for $\zeta$ close to $\zeta^{\ast}.$ Therefore the limit
$\lim_{\zeta\rightarrow\zeta^{\ast}}G\left(  \zeta\right)  $ exists and:%
\[
\lim_{\zeta\rightarrow\zeta^{\ast}}G\left(  \zeta\right)  \geq\frac{1}{3}.%
\]

Since $G$ is monotonically increasing we can parametrize $Z$ as a function of 
$G.$ Let us denote the corresponding function by 
$Z=\tilde{Z}\left(  G\right)  .$ Then by (\ref{ED02}) and (\ref{ED03}):
\begin{equation}
\frac{d (\log\tilde Z)}{dG}
=\frac{(3G-1-\delta_1(\zeta))(1+\tilde Z^{-2})}
{2G[(1-G)-(1+\tilde Z^{-1})\delta_2(\zeta)]}.
\end{equation}
If the limit of $G$ were less than one the right hand side of this expression
would be bounded and it would follow that $Z$ was bounded, a contradiction.
Hence $\lim_{\zeta\to\zeta^*}G(\zeta)=1$.

To complete the proof the condition on the initial data in the hypotheses of
the lemma will be used. Since $\lim_{\zeta
\rightarrow\zeta^{\ast}}Z\left(  \zeta\right)  =\infty,\ \lim_{\zeta
\rightarrow\zeta^{\ast}}q\left(  \zeta\right)  =0,$ $\lim_{\zeta
\rightarrow\zeta^{\ast}}\left(  Zq\right)  \left(  \zeta\right)  =0$ and
$\lim_{\zeta\rightarrow\zeta^{\ast}}G\left(  \zeta\right)  >\frac{1}{3}$ it
follows that for any sufficiently small $\delta>0$ and for any solution 
$\left(  \bar{q},\bar{G},\bar
{Z}\right)  $ that is sufficiently close  to $\left(
q,G,Z\right)  $ at $\zeta=0$ we have for some $\zeta_{0}<\zeta^{\ast}:$%
\begin{equation}
\bar{q}\left(  \zeta_{0}\right)  \leq\delta^3\ \ ,\ \ \bar{G}\left(  \zeta
_{0}\right)  \geq\frac{1}{3}+\delta\ \ ,\ \ \left(  \bar{Z}\bar{q}\right)
\left(  \zeta_{0}\right)  \leq\delta\ \ ,\ \ \bar{Z}\left(  \zeta_{0}\right)
\geq\frac{1}{\delta}.\label{ine1}%
\end{equation}
It will now be shown that for $\delta$ sufficiently small the region defined 
by these four inequalities is invariant. On the part of the boundary of the 
region where $\bar q=\delta^3$ we have $\frac{d \bar q}{d\zeta}<0$. On the 
part of the boundary where $\bar Z\bar q=\delta$ assuming that 
$\delta<3^{-\frac13}$ suffices to show, using (\ref{dZq1}), 
that the derivative of $\bar Z\bar q$ is negative. On the part
with $\bar G=\frac13+\delta$ the following inequality holds:
\begin{equation}
\frac{\partial\bar G}{\partial\zeta}\ge\frac23\left[\frac{2}{3\delta}-1
-(\delta^{\frac12}+\delta^{\frac52})^2-(\delta+\delta^3)^2\right].
\end{equation}
Choosing $\delta$ sufficiently small implies that the right hand side of this
inequality is positive. On the whole region
\begin{equation}
\frac{d\bar Z}{d\zeta}\ge\delta(3-\sqrt{\delta^2+\delta^6}).
\end{equation}
If $\delta$ is small enough then this quantity is positive.
%It then follows from (\ref{EDO1}) that $\bar{q}\left(  \zeta\right)
%\leq\delta$ for $\zeta\geq\zeta_{0}$ as long as $\bar Z$ remains positive.
%Moreover, using (\ref{dZq1}) it follows
%that $\left(  \bar{Z}\bar{q}\right)  \left(  \zeta\right)  \leq\left(
%2\right)  ^{\frac{1}{4}}\delta^{\frac{1}{2}}$ for $\zeta\geq\zeta_{0}$  Then
%(\ref{ED02}) implies that $\bar G$ is increasing if $\delta$ is chosen 
%sufficiently small. Using (\ref{ine1}) 
Putting these facts together shows that the solution starts in the region of 
interest when $\zeta=\zeta_0$ and stays there. In particular
$\bar{G}\left(  \zeta\right)\geq\frac{1}{3}+\delta$ for $\zeta\geq\zeta_{0}.$ 
Therefore $\bar{Z}$ blows up in finite time due to (\ref{ED03}) and Lemma 
\ref{LeT3} follows. 
\end{proof}

\bigskip

\noindent
\textit{Acknowledgements}: JJLV is grateful to J. M.
Mart\'{\i}n-Garc\'{\i}a for interesting discussions concerning the analogies
and differences between the solutions of this paper and those in
\cite{martingarcia}. JJLV acknowledges support of the Humboldt Foundation, the
Max Planck Institute for Gravitational Physics (Golm), the Max Planck
Institute for Mathematics in the Sciences (Leipzig), the Humboldt University
in Berlin and DGES Grant MTM2007-61755 and Universidad Complutense. Both
authors are grateful to the Erwin Schr\"{o}dinger Institute in Vienna, where
part of this research was carried out, for support.

\end{document}